\documentclass[11pt,oneside]{article}

\usepackage[round,colon,authoryear]{natbib}

\usepackage{graphicx}
\usepackage{amsmath}
\usepackage{amsfonts}
\usepackage{amssymb}
\usepackage{amsthm}
 \usepackage{mathrsfs}
\usepackage{enumerate}
\usepackage{pdfsync}
 \usepackage{stmaryrd}
\usepackage{comment}
\usepackage{color}
\usepackage{bm}
\usepackage{times}
\usepackage{enumitem}

\usepackage{bm}
\usepackage{tikz}

\bmdefine\taub{\tau}
\bmdefine\mub{\mu}
\bmdefine\lab{\lambda}
\bmdefine\varsigmab{\varsigma}

 \numberwithin{equation}{section}
    
\newcommand{\R}{\mathbb{R}}
\newcommand{\N}{\mathbb{N}}
\newcommand{\1}{\mathbf{1}}

\newcommand{\Q}{\Qu}

\newcommand{\dx}{\mathrm{d}}

\renewcommand{\P}{\textsf{\upshape P}}
\newcommand{\Qu}{\textsf{\upshape Q}}
\newcommand{\filt}[1]{\mathfrak{#1}}
\newcommand{\sigalg}[1]{\mathscr{#1}}

\newcommand{\lc}{[\![}

\newtheorem{thm}{Theorem}[section]
\newtheorem{lem}[thm]{Lemma}
\newtheorem{prop}[thm]{Proposition}
\newtheorem{cor}[thm]{Corollary}

\theoremstyle{definition}
\newtheorem{defn}[thm]{Definition}

\theoremstyle{remark}
\newtheorem{example}[thm]{Example}

\theoremstyle{remark}
\newtheorem{rem}[thm]{Remark}

\title{
Volatility and Arbitrage
\thanks{We are grateful to Adrian Banner,  David Hobson, Blanka Horv\'ath, Antoine Jacquier, Vassilios Papathanakos, Josef Teichmann, Minghan Yan, and Alexander Vervuurt  for helpful comments. I.K.~acknowledges   support from  the National Science Foundation under Grant   NSF-DMS-14-05210.   J.R.~acknowledges   generous   support  from  the Oxford-Man Institute of Quantitative Finance, University of Oxford.}
 }

\author{  
\textsc{E. Robert Fernholz}  \thanks{
  \textsc{Intech} Investment Management,  One Palmer Square, Suite 441, Princeton, NJ 08542      (E-mail:    {\it bob@bobfernholz.com}). 
  }  
 \and
\textsc{Ioannis Karatzas}                \thanks{ 
Department of Mathematics, Columbia University (E-mail:    {\it ik1@columbia.edu}); and \textsc{Intech} Investment Management,  One Palmer Square, Suite 441, Princeton, NJ 08542      (E-mail:    {\it ikaratzas@intechjanus.com}).
  }  
 \and
\textsc{Johannes Ruf}                \thanks{ 
Department of Mathematics, University College London, Gower Street, London WC1E 6BT, United Kingdom (E-mail:    {\it j.ruf@ucl.ac.uk}).
          }
                                      }

\begin{document}

\maketitle

\begin{abstract} \noindent
 The capitalization-weighted total relative variation $\sum_{i=1}^d \int_0^\cdot   \mu_i (t)  \dx \langle \log \mu_i \rangle (t)$  in an equity market consisting of a fixed number $d$ of assets  with capitalization weights $\mu_i (\cdot)$  is an observable  and nondecreasing function of time. If this observable of the market is not just nondecreasing, but actually grows at a rate which is  bounded away from zero, then strong arbitrage can be constructed relative to the market  over sufficiently long time horizons.  It has been an open issue for more than ten years, whether such strong outperformance of the market is possible also over  arbitrary time horizons under the stated  condition. We show that this is not possible in general, thus settling this long-open question. We also show that, under appropriate additional conditions, outperformance over any time horizon indeed becomes possible, and exhibit   investment strategies that effect it. 
\end{abstract}

\smallskip
\noindent{\it Keywords and Phrases:} Trading strategies, functional generation, relative arbitrage, short-time arbitrage,   support of diffusions, diffusions on manifolds, nondegeneracy.


\input amssym.def
\input amssym

\section{Introduction and summary}

At least since \cite{Fe}, it has been known that  volatility in a stock market can generate arbitrage, or at least {\em relative arbitrage} between a specified portfolio and the market portfolio. However, the questions of exactly what level of  volatility is required, and how long it might take, for this arbitrage to be realized,  have never been fully answered. Here we hope to shed some light on these questions and come to an understanding about what might represent adequate volatility, and over which time-frame relative arbitrage might be achieved.

A common condition regarding market volatility,  sometimes known as {\em strict nondegeneracy}, is the requirement that the eigenvalues of the market covariation matrix be bounded away from zero. It was shown in \cite{Fe} that strict nondegeneracy, coupled with market {\em diversity}, the condition that the largest relative market weight be bounded away from one, will produce relative arbitrage with respect to the market over   sufficiently long  time horizons. Later,  \cite{FKK} showed that these   two conditions lead to relative arbitrage over arbitrarily short time horizons. Market diversity is actually a rather mild condition, one that would be satisfied in any market with even a semblance of anti-trust regulation. However, strict nondegeneracy is a much stronger condition, and   probably not amenable to statistical verification in any realistic market setting. While it might be reasonable to assume that the market covariation matrix is {\em nonsingular}, it would seem rather courageous to make strong assumptions regarding the  behavior over time of the smallest eigenvalue of a random $d\times d$ matrix, where $d \in \N$ is usually a large integer, standing for the number of stocks in an equity market. 

Accordingly,  it is preferable to avoid the use of strict nondegeneracy as a characterization of adequate volatility, and  
consider instead measures based on aggregated relative variations. The most important such measure is the so-called {\em cumulative excess growth} $ \Gamma^{\bm H} (\cdot)$. This measure is based on the weighted average of the variances of the logarithmic market weights. The weights  used are exactly the market weights; more precisely, we have
\begin{equation}\label{eq:1.1}
\Gamma^{\bm H} (\cdot) := 
\frac{1}{2}   \sum_{i=1}^d \int_0^{ \cdot}  \mu_i (t)   \, \mathrm{d}    \big \langle \log \mu_i    \big  \rangle (t) .
\end{equation}
This quantity will be discussed at some length below. Here, $\mu_i(t)$ represents the market weight of the $i$-th stock at time $t \geq 0$, for each $i = 1, \cdots, d$. \cite{FKVolStabMarkets} show that if the slope of $\Gamma^{\bm H} (\cdot) $ is bounded away from zero, then relative arbitrage with respect to the market  will exist over a long enough time period. We shall see  in Section~\ref{S:6} that this condition does not necessarily imply relative arbitrage over an arbitrarily short period of time. However,  not all is lost: in Section~\ref{S:5} we shall see that under {\it  additional} assumptions, such as when there are only two stocks $(d=2)$ or when the market weights satisfy  appropriate  time-homogeneity properties, relative arbitrage {\it does} exist  over arbitrary time horizons.  Other  sufficient conditions are also provided.
We remark  also that \cite{Pal:exponentially} recently derived  sufficient conditions for large markets that yield asymptotic short-term arbitrage.

\smallskip
\noindent
{\it Preview:} The structure of the paper is as follows: Sections~\ref{S:2}, \ref{S:3}, and \ref{S:4} introduce the basic definitions, including the concept of generating functions. Introduced by \cite{F_generating, Fe} and developed in \cite{FK_survey} and \cite{Karatzas:Ruf:2016}, these functions are useful in creating trading strategies that produce arbitrage relative to the market. Section~\ref{S:5} establishes conditions under which relative arbitrage can be shown to exist over arbitrary time horizons.   Section~\ref{S:6} constructs examples of markets which have adequate volatility but  no arbitrage --- indeed, the price processes in this examples are all martingales. Section~\ref{S:7} summarizes the results of this paper and discusses some open questions.

\smallskip
 One final note before we begin, and this regards {\em classical arbitrage} as compared to {\em relative arbitrage}. Classical arbitrage is measured versus cash, which in our context can be considered to be a constant, positive process. The relative arbitrage we present here is generally versus the market portfolio, and this means that the market portfolio replaces cash as the benchmark against which the relative arbitrage is measured. All types of arbitrage have a boundedness restriction to prevent ``doubling'' strategies. For classical arbitrage, the value of the trading strategy that creates the arbitrage must be bounded from below relative to cash. In relative arbitrage, this bound is relative to the market. In general, there is no bound on market value, so one bound does not yield the other one, and vice versa. Hence, these two types of arbitrage can be incompatible.

\section{The market model and trading strategies} 
 \label{S:2}

We fix a probability space $(\Omega, \sigalg{F}, \P)$, endowed  with a right-continuous filtration $ \filt{F} = (\sigalg{F} (t) )_{t \geq 0 }$. For simplicity, we take  $\sigalg{F}(0) = \{ \emptyset, \Omega \}$, mod.~$\P$. All processes to be  encountered  will be adapted to this filtration. On this filtered probability space and for some $d \in \N \setminus \{1\}$, we consider 
a continuous $d$--dimensional semimartingale $\mu(\cdot) = (\mu_1(\cdot), \cdots, \mu_d(\cdot))'$ taking values in the   lateral face of the unit simplex
\begin{equation*}
{\bm \Delta}^d := \bigg\{ \big(x_1, \cdots, x_d \big)'  \in [0, 1]^d\,:\, \sum_{i=1}^d x_i =1 \bigg\}  \subset \mathbb H^d  ,
\end{equation*}
 where ${\mathbb{H}}^d$ denotes the hyperplane   
\begin{equation}
\label{eq:160510.1}
{\mathbb H}^d := \bigg\{ \big(x_1, \cdots, x_d \big)'  \in \R^d\,:\, \sum_{i=1}^d x_i =1 \bigg\}. 
\end{equation}
We assume that $\mu(0) \in \bm \Delta^d_+$, where we set
\begin{align} \label{eq:160716.1}
	\bm \Delta^d_+ := \bm \Delta^d \cap (0,1)^d.
\end{align}
We interpret $\mu_i(t)$ as the relative weight, in terms of capitalization in the market, of company $i = 1, \cdots, d$ at time $t \geq 0$.  An individual company's weight is allowed to become zero, but we insist that $ \sum_{i=1}^d \mu_i (t) = 1$ must hold for all $t \geq 0$.

 In this spirit, it is useful   to think of the generic market weight process $\mu_i (\cdot)$ as the ratio
\begin{equation}
\label{eq: 2.1a}
\mu_i (\cdot)    := \frac{S_i (\cdot)}{ \Sigma (\cdot)}, \qquad i=1, \cdots, d;\qquad \Sigma (\cdot) :=  S_1 (\cdot) + \cdots + S_d (\cdot) >0. 
\end{equation}
Here $S_i (\cdot)$ is a continuous nonnegative semimartingale  for each $i = 1, \cdots, d$, representing the capitalization (stock-price, multiplied by the number of shares outstanding) of the $i$-th company; whereas the  process $\Sigma (\cdot)$, assumed to be strictly positive, stands of the total capitalization of the entire market.

For later reference, let us introduce the stopping times
\begin{equation}
\label{eq: 6}
\mathscr D := \mathscr D_1 \wedge \cdots \wedge \mathscr D_d, \qquad \mathscr D_i : = \inf \big\{ t \ge 0 :\, \mu_i (t) = 0 \big\}.
\end{equation}
To avoid notational inconveniences below, we assume that $\mu_i(\mathscr D_i +t) = 0$ holds for all $i = 1, \cdots, d$ and $t \geq 0$; in other words, zero is an absorbing state for any of the market weights.

One of our results, Theorem~\ref{T:5.2} below, needs the following notion.
\begin{defn}[Deflator]  \label{D:deflator}
A strictly positive process $Z(\cdot)$ is called {\it deflator} for the vector semimartingale $\mu (\cdot)$ of relative market weights, if the product 
$Z(\cdot) \mu_i (\cdot)$ is a local martingale  for every $i=1, \cdots, d$.  
\end{defn}
 Except when explicitly stated otherwise, the results below will hold  independently of whether the market model admits  a deflator, or not.

We  now consider   a predictable process 
$
\vartheta ( \cdot) = \big( \vartheta_1 ( \cdot) , \cdots, \vartheta_d ( \cdot) \big)^\prime
$
with values in $ \R^d$, and interpret   $ \vartheta_i ( t)$ as the number of shares held at time $t \geq 0$ in the stock of company $ i=1, \cdots, d$.  Then  the total {\it value}, or ``wealth'',  of this investment, measured in terms of the total market capitalization,
 is
\begin{equation*}
V^\vartheta (\cdot):= \sum_{i=1}^d  \vartheta_i ( \cdot) \mu_i (\cdot).
\end{equation*}

\begin{defn} [Trading strategies]
\label{def: TS}
Suppose that the $ \R^d$--valued, predictable process $ \vartheta (\cdot)$ is integrable with respect to the continuous semimartingale $\mu( \cdot) $. We shall say that $ \vartheta (\cdot)$ is a  {\it trading strategy} if it satisfies the so-called ``self-financibility'' condition
\begin{equation}
\label{eq: SF}
V^\vartheta (T) =V^\vartheta (0) +  \int_0^T \sum_{i=1}^d  \vartheta_i ( t) \mathrm{d} \mu_i (t), \qquad T \geq 0.
\end{equation}
We call a trading strategy $\vartheta(\cdot)$ {\it long-only,} if it never sells any stock short: i.e., if $  \vartheta_i ( \cdot) \geq 0$   holds   for all $i = 1, \cdots, d$. 
   \end{defn}

The vector stochastic integral in  \eqref{eq: SF} gives the {\it gains-from-trade} realized over   $[0,T]$. The self-financibility requirement of \eqref{eq: SF} posits that these ``gains'' account for the entire change in the value generated by the trading strategy $ \vartheta (\cdot)$ between the start $ t=0$ and the end $ t=T$ of the   interval $[0,T]$.

The following observation, a result of straightforward computation,  is needed in  the proof of Theorem~\ref{T:5.1}.
\begin{rem}[Concatenation of trading strategies]
\label{R:2.4}
Suppose we are given a real number $b \in \R$, a stopping time $\tau$, and a trading strategy $ \varphi (\cdot) $. We then form a new process $\psi (\cdot) = \big( \psi_1 (\cdot), \cdots,  \psi_d (\cdot) \big)'$, again integrable with respect to $\mu(\cdot)$ and with components
\begin{equation} \label{eq:160717.3}
\psi_i (\cdot) := b  + \big( \varphi_i (\cdot)- V^\varphi (\tau) \big) \1_{\lc\tau, \infty\lc} (\cdot), \qquad i = 1, \cdots, d.
\end{equation}
Then the process $\psi (\cdot)$ is a trading strategy itself, and its associated wealth process $V^\psi (\cdot) $ is given by 
\begin{equation*}
V^\psi (\cdot) = b  + \big(V^\varphi  ( \cdot) - V^\varphi (\tau) \big) \1_{\lc\tau, \infty\lc} (\cdot).\end{equation*}
\end{rem}

\section{Functional generation of trading strategies} 
 \label{S:3}

There is a special class of trading strategies, for which the representation of \eqref{eq: SF} takes an exceptionally simple and explicit form; in particular, one in which stochastic integrals disappear entirely from the right-hand side of \eqref{eq: SF}.  
In order to present this class of trading strategies, we start with a {\it regular function:}   a continuous mapping $\bm G : \bm \Delta^d \to \R$ that satisfies a generalized It\^o rule. 
By this, we mean that the process $\bm G(\mu(\cdot))$ can be written as the sum 
\begin{equation}
\label{eq: 3.0}
\bm G(\mu (\cdot)) = \bm G(\mu (0))+ \int_0^\cdot \sum_{i=1}^d D_i  \bm G(\mu ( t))\, \mathrm{d} \mu_i (t) - \Gamma^{\bm G} (\cdot)
\end{equation}
of a constant initial condition $\bm G(\mu(0))$, of a stochastic integral with respect to $\mu(\cdot)$ of another measurable function $D\bm G: \bm \Delta^d \to \R^d$ evaluated at $\mu(\cdot)$,  and of a process $-\Gamma^{\bm G}(\cdot)$ which has finite variation on compact time-intervals.  The precise definition can be found in  \cite{Karatzas:Ruf:2016}. We  stress here that the notion of regular function is relative to a given   market weight process $\mu(\cdot)$; a  function $\bm G$ might be regular with respect so some such process, but not with respect to another.

In this paper, we shall only consider regular functions $\bm G$ that can be extended to   twice continuously differentiable functions in a neighborhood of the set $\bm \Delta^d_+$  in \eqref{eq:160716.1}. From now on, every regular function $\bm G$   we encounter will be  supposed to have this smoothness property. Then we may assume that $D{\bm G}(x)$ is the gradient of $\bm G$ evaluated at $x$, at least for all $x \in \bm \Delta^d_+ $. Moreover, the  finite-variation process $\Gamma^{\bm G}(\cdot)$  will then be  given, on the stochastic interval $\lc0, \mathscr D\lc$ and in the notion of \eqref{eq: 6}, by the expression 
\begin{equation}
\label{eq: Gamma}
  \Gamma^{\bm G} (\cdot)= -
\frac{1}{2}  \sum_{i=1}^d \sum_{j=1}^d \int_0^{ \cdot}    D^2_{i,j}
{\bm G}\big(\mu (t)  \big)   \,  \mathrm{d} \big \langle \mu_i, \mu_j  \big \rangle (t).
     \end{equation}

There are two ways in which a regular function $\bm G$ can generate a trading strategy.
\begin{enumerate}
	\item {\it Additive generation:} The vector process $\bm \varphi^{\bm G} (\cdot) = \big( \bm \varphi_1^{\bm G} (\cdot), \cdots,  \bm \varphi_d^{\bm G} (\cdot) \big)'$ with components
\begin{equation}
\label{eq:3.2}
{\bm \varphi}^{\bm G}_i ( \cdot) := D_i {\bm G}\big(\mu (\cdot) \big)+ \Gamma^{\bm G} (\cdot) +    {\bm G}\big(\mu (\cdot) \big)   -  \sum_{j=1}^d  \mu_j (\cdot) D_j {\bm G}\big(\mu (\cdot) \big)  , \qquad i=1, \cdots, d,
\end{equation}
is a trading strategy, and is said to be {\it  additively  generated by $\bm G$.}
The  wealth process $ V^{{\bm \varphi}^{\bm G}} (\cdot ) $ associated to this trading strategy has the extremely simple form
\begin{equation}
\label{eq: valphiG}
V^{{\bm \varphi}^{\bm G}} (\cdot ) = \bm {\bm G}(\mu (\cdot) ) + \Gamma^{\bm G} (\cdot).
\end{equation}
That is, $V^{{\bm \varphi}^{\bm G}} (\cdot ) $ can be represented as the sum of a completely observed and ``controlled" term $\bm G(\mu ( \cdot))$, plus a ``cumulative  earnings" term $\Gamma^{\bm G} (\cdot) $. It is important to note that this expression is completely devoid of stochastic integrals.
	\item  {\it Multiplicative generation:}  The second way to generate a trading strategy requires the process $1/\bm G(\mu(\cdot))$ to be locally bounded.  This assumption allows us to define the process  
 \begin{equation}
\label{eq: 3.7a}
 Z^{\bm G} (\cdot) := \bm G  \big( \mu (\cdot)  \big)     \exp \left(    \int_0^{\cdot} \frac{\mathrm{d} \Gamma^{\bm G}   (t)    }{  \bm G  \big( \mu (t)\big)   }\right) > 0.
  \end{equation} 
  Then the vector process $\bm \psi^{\bm G} (\cdot) = \big( \bm \psi_1^{\bm G} (\cdot), \cdots,  \bm \psi_d^{\bm G} (\cdot) \big)'$ with components
\begin{equation}
\label{eq: 3.7b}
\bm \psi^{\bm G}_i (\cdot) :=  Z^{\bm G}  (\cdot   ) \bigg(  1 + \frac{1   
}{\bm  G \big( \mu ( \cdot) \big) } \bigg( D_i 
\bm G \big( \mu ( \cdot) \big)  - \sum_{j=1}^d D_j   
\bm G \big( \mu (  \cdot) \big)  \mu_j ( \cdot)  \bigg)   \bigg), \qquad i=1, \cdots, d
\end{equation} 
is a trading strategy, and is said to be {\it multiplicatively generated by} $\bm G$. The value $ V^{ {\bm \psi}^{\bm G}}     (\cdot)$  this strategy generates,  is given by the  process   of \eqref{eq: 3.7a}, namely: $ V^{ {\bm \psi}^{\bm G}}  (\cdot)= Z^{\bm G}   (\cdot)$. 
\end{enumerate}

If the regular function $\bm G$ is concave, then it can be checked that both strategies ${\bm \varphi}^{\bm G}  ( \cdot)  $  and  ${\bm \psi}^{\bm G}  ( \cdot)  $ are long-only. Moreover, in this case, the process   $\Gamma^{\bm G} (\cdot)$   in \eqref{eq: 3.0} is actually nondecreasing.  More generally, we introduce the following notion.

\begin{defn}[Lyapunov function]  
A regular function $\bm G$ is called    {\it Lyapunov function}, if the 
process  $\Gamma^{\bm G} (\cdot)$  in \eqref{eq: 3.0}  is nondecreasing.  
\end{defn}

 For a Lyapunov function $\bm G$ the process $\Gamma^{\bm G} (\cdot)$ has    the significance of an aggregated {\it measure of cumulative   volatility} in the market; the Hessian matrix-valued process $-D^2 {\bm G} (\mu (\cdot))$, which effects the aggregation,  acts then as a sort of ``local curvature" on the covariation matrix of the semimartingale $\mu (\cdot),$ to give us this cumulative measure of volatility.

The theory and applications of functional  generation were developed by  \citet{F_generating, Fe}; see also \cite{Karatzas:Ruf:2016}. All   claims made in this section are proved in these references, and several  examples of functionally generated trading strategies are   discussed. We  introduce now  four regular functions which will be important here.

\begin{enumerate}
	\item The entropy function of statistical mechanics and information theory
\begin{equation}
\label{eq: 3.13}
\bm H := - \sum_{i=1}^d x_i  \log  x_i  , \qquad x \in {\bm \Delta}^d
\end{equation}
with the convention $0 \times \log \infty = 0$
	is a particularly important 
	regular function. 	Note that $\bm H$ is concave, thus also a Lyapunov function, and takes values in $[0,  \log d]$. It generates additively the long-only {\it entropy-weighted trading strategy}
\begin{equation}
\label{eq: 3.14}
{\bm \varphi}^{\bm H}_i(\cdot) =\left(  \log \left( \frac{1} {\mu_i ( \cdot) }\right) + \Gamma^{\bm H} ( \cdot)  \right) \1_{\{  \mu_i(\cdot)>0\}},   \qquad i=1, \cdots, d.
\end{equation}
Here
 \begin{equation}
 \label{eq: 3.1}
\Gamma^{\bm H} (\cdot) = 
\frac{1}{2}   \sum_{i=1}^d \int_0^{ \cdot} \frac{    \mathrm{d}    \big \langle  \mu_i   \big  \rangle  (t)}{\mu_i (t)}  =
\frac{1}{2}   \sum_{i=1}^d \int_0^{ \cdot}  \mu_i (t)   \, \mathrm{d}    \big \langle \log  \mu_i     \big  \rangle (t) \qquad \text{on $\lc0, \mathscr D\lc$}
\end{equation}
denotes the cumulative earnings   of the strategy ${\bm \varphi}^{\bm H} (\cdot)$,  as well as the ${\bm H}-$aggregated   measure of cumulative   volatility in the market  in the manner of \eqref{eq: 3.0}.   This nondecreasing, trace-like process $\Gamma^{\bm H} (\cdot)$ has already been encountered in \eqref{eq:1.1}; it is   called the {\it cumulative excess growth} of the market in Stochastic Portfolio Theory, and plays an important role there. 

The process $\Gamma^{\bm H} (\cdot)$ measures the market's cumulative total relative variation   -- stock-by-stock, then averaged according to each stock's   weight. As such, it offers a gauge of the market's ``intrinsic volatility.''  Figure~1   uses the monthly stock database of the Center for Research in Securities Prices (CRSP) at the University of Chicago, to plot  the quantity 
of \eqref{eq: 3.14} over the 80-year period 1926-2005.

\begin{figure}[!htbp] 
				\centering
				\includegraphics[width=0.70 \textwidth]{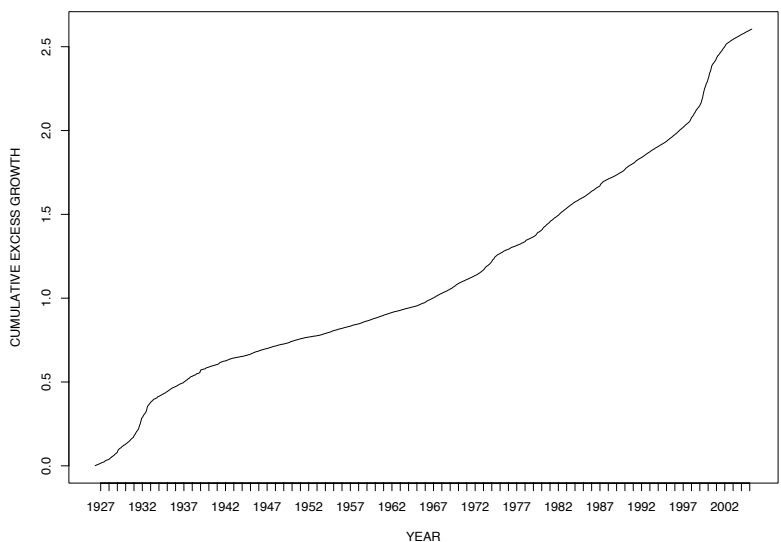}\
				\caption{{\small Cumulative intrinsic variation  
				 $\Gamma^{\bm H} (\cdot) $ 
				 for the U.S. Market, 1926-2005.}}
			\end{figure}

	\item The quadratic 
\begin{align} 
\label{eq:Q}
	\bm Q(x) := 1 - \sum_{i = 1}^d x_i^2, \qquad x \in \bm \Delta^d,
\end{align}
takes values in $[0, 1-1/d]$, and is also a concave regular function. It is mathematically very convenient to work with $\bm Q$, so  this function will play a major role when constructing specific counterexamples in Section~\ref{S:6}.  
The corresponding aggregated   measure of cumulative   volatility is   given by the nondecreasing  trace process
\begin{align}\label{eq:GQ}
	\Gamma^{\bm Q} (\cdot) = \sum_{i=1}^d \langle \mu_i\rangle(\cdot);
\end{align}
we note also that the difference  $ 2 \Gamma^{\bm H} (\cdot) - \Gamma^{\bm Q}(\cdot)$ is nondecreasing, where $ \Gamma^{\bm H}(\cdot)$ is given in \eqref{eq: 3.1}.
\item We shall also have a close  look at the concave, geometric mean  function 
\begin{align} \label{eq:R}
\bm R(x) := \bigg(  \prod_{i=1}^d x_i \bigg)^{1/d}, \qquad x \in  \bm \Delta^d.
\end{align}

\item Finally, for $q \geq 1$, we shall use, in one of the proofs, the power 
 	\begin{align} \label{eq:P}
		\bm F(x) := x_1^q, \qquad x \in \bm \Delta^d.
	\end{align}
		Note that $\bm F$ is, for $q>1$, a {\it convex} rather than concave function, as was the case in the other three examples. It is nevertheless   regular, so it can still be used as a generating function.  Indeed, if $1/\mu_1(\cdot)$ is locally bounded away from zero, the multiplicatively generated strategy $\bm \psi(\cdot)$ of \eqref{eq: 3.7b} exists. More precisely, with the process of \eqref{eq: 3.7a} given now by
\begin{align}  \label{eq:160717.2}
Z^{\bf F} (\cdot) =  \big( \mu_1 (\cdot) 
\big)^{q }  \exp \left( - \frac{1}{2} q\big( q-1 \big) \int_0^{\cdot} \big( \mu_1 (t) 
\big)^{-2}  \dx \langle \mu_1 \rangle  (t)\right),
	\end{align}
the expression in \eqref{eq: 3.7b} can be written here as
\begin{align}  \label{eq:160717.1}
\bm \psi_1^{\bf F}  (\cdot) = \left( \frac{q}{\mu_1(\cdot)} +1-q \right)   Z^{\bf F}  (\cdot); \qquad 
\bm \psi_i^{\bf F}  (\cdot) = \big( 1-q \big)  Z^{\bf F}  (\cdot), \quad i=2, \cdots, d.
\end{align}
\end{enumerate}

\section{Relative arbitrage, and an old question} 
 \label{S:4}

We introduce now the important notion of relative arbitrage with respect to the market. 
\begin{defn}[Relative arbitrage]
\label{D:4.1}
Given a real constant $T>0$, we say that a trading strategy $\vartheta (\cdot) $ is a {\it relative arbitrage with respect to the market} over the time horizon $[0,T]$ if $V^\vartheta (0) = 1$, $V^\vartheta(\cdot) \geq 0$, and
\begin{equation*}
\P \big( V^\vartheta (T) \ge 1 \big) =1, \qquad \P \big( V^\vartheta (T) >1 \big) >0.
\end{equation*}
If in fact $\P \big( V^\vartheta (T) > 1 \big) =1$ holds,   this  relative arbitrage is called {\it strong}.
\end{defn}

\begin{rem}[Equivalent martingale measure]\label{R:4.2}
Fix a real number $T > 0$. Then no relative arbitrage is possible over the time horizon $[0,T]$  with respect to a market whose relative weights $ \mu_1 (\cdot \wedge T), \cdots, \mu_d (\cdot \wedge T)$ are martingales under some equivalent probability measure $\Q_T \sim \P$  defined on $\sigalg F (T)$. 

Suppose now that no relative arbitrage is  possible over the  time horizon $[0,T]$, with respect to a market with relative weights $ \mu (\cdot)$. Provided that a deflator for $\mu (\cdot)$ exists, an  equivalent probability measure $\Q_T \sim \P$ then exists on $\sigalg{F} (T)$, under which the relative weights $\mu_1 (\cdot \wedge T), \cdots, \mu_d (\cdot \wedge T)$ are martingales;
see \cite{DS_fundamental} or \cite{KK}.
\end{rem}

 Since the   process $\mu(\cdot)$   expresses   
 the market portfolio, the arbitrage of Definition~\ref{D:4.1} can be interpreted as   {\it relative arbitrage with respect to the market.}   The question of whether a given market portfolio can be ``outperformed" as in   Definition~\ref{D:4.1}, is of great theoretical and practical importance -- particularly   given the proliferation of index-type mutual funds that try to track  and possibly outperform a specific  benchmark  market portfolio (or ``index"). To wit:    
 {\it Under what conditions is there relative arbitrage with respect to a specific market portfolio? over which time horizons? if it exists, can such    relative arbitrage  be strong?}

 Functionally-generated trading strategies are ideal for answering such questions, thanks to the representations of \eqref{eq: Gamma} and \eqref{eq: valphiG}  which describe   their performance relative to the market  in a pathwise manner, devoid of stochastic integration. The following result is taken from  \cite{Karatzas:Ruf:2016}; its lineage goes back to  \cite{Fe} and  to \cite{FKVolStabMarkets}. In our present context, it is a straightforward consequence of the representation  \eqref{eq: valphiG}. 
 
\begin{thm}[Strong relative arbitrage over sufficiently long  time horizons]
 \label{T:4.3}
Suppose that $\bm G: \bm \Delta^d \rightarrow [0,\infty)$ is a Lyapunov function with $\bm G (\mu (0))>0$. Suppose, moreover, that there is a  real number  $ T_* >0$  with  the property 
\begin{equation}   
 \label{eq:160131.3}
\P \big( \Gamma^{\bm G} (T_*) > \bm G(\mu(0)) \big) = 1 .
\end{equation}
  Then the trading strategy $ {\bm \varphi}^{\bm G_*}  ( \cdot) = \big(   {\bm \varphi}^{\bm G_*}_1  ( \cdot), \cdots,   {\bm \varphi}^{\bm G_*}_d  ( \cdot) \big)'$, generated additively in the manner of \eqref{eq:3.2} by the function $\bm G_* := \bm G / \bm G (\mu (0))$, is strong 
 relative arbitrage with respect to the market  over any time horizon $[0,T] $ with $T \geq T_*$.   
 \end{thm} 
 
  The following  is a direct  corollary of  Theorem~\ref{T:4.3}.
 \begin{cor}[Slope bounded from below]  \label{C:4.4}
Suppose that $\bm G: \bm \Delta^d \rightarrow [0,\infty)$ is a regular function with $\bm G (\mu (0))>0$ such that
\begin{equation}
\label{eq: 4.3}
\P \left( \text{the mapping} ~ [0,\infty) \ni t \mapsto \Gamma^{\bm G}(t)- \eta t \,\,\,\text{is nondecreasing} \right)= 1, \quad \text{for some  $\eta >0$}.
\end{equation}
Then the trading strategy ${\bm \varphi}^{\bm G_*}(\cdot)$ of Theorem~\ref{T:4.3} is strong relative arbitrage with respect to the market, over any time horizon $[0,T]$ with
\begin{equation}
\label{eq: 4.4}
T > \frac{ {\bm G}(\mu (0) )}{\eta}.
\end{equation}
\end{cor}

The assertion of Corollary~\ref{C:4.4}  appears already  in \cite{FKVolStabMarkets} for the entropy function $\bm H$ of \eqref{eq: 3.13}.  In this case, the condition of \eqref{eq: 4.3} posits 
that the cumulative relative variation $\Gamma^{\bm H} (\cdot) $ as in \eqref{eq: 3.1} is not just increasing, but actually dominates a straight line with positive slope. 

   This assumption can be read most instructively in conjunction with the plot of  Figure~1.  Under it, Corollary~\ref{C:4.4} guarantees  the existence of relative arbitrage with respect to the market  over any time horizon $[0,T]$ of duration $T > \bm H(\mu(0)) / \eta$.

\begin{rem}
  \label{R:old}
The following question was posed in  \cite{FKVolStabMarkets}, and was asked again in \cite{Banner_Fernholz}.  
{\it Assume that \eqref{eq: 4.3} holds with $\bm G = \bm H$, given in \eqref{eq: 3.13}. 
Is then relative  arbitrage with respect to the market  possible over every time horizon $[0, T]$  with arbitrary length $T>0$?}

In  Section~\ref{S:5}  we shall present results which guarantee, under appropriate {\it additional}  conditions,   affirmative answers to this question.
Then in   Section~\ref{S:6} we shall construct market models illustrating that, in general, the answer to the above question is negative. This settles an issue which had remained open for more than ten years.
 \end{rem}

\section{Existence of short-term relative arbitrage} 
 \label{S:5}

Given a {Lyapunov} function $\bm G: \bm \Delta^d \to [0, \infty)$   with $\bm G (\mu (0))>0$,    Theorem~\ref{T:4.3}  provides the condition $\P\big( \Gamma^{\bm G} (T ) > \bm G (\mu(0))   \big)=1$      on the length $T > 0 $ of the time horizon $[0,T]$  as  sufficient for the existence of strong relative arbitrage with respect to the market  over this time horizon. 

In this section, we study conditions under which relative arbitrage exists on the time horizon $[0,T]$, for any real number $T > 0$. Subsection~\ref{SS:5.1} discusses conditions that guarantee the existence of strong short-term relative arbitrage. The conditions of Subsection~\ref{SS:5.2} guarantee only the existence of  short-term relative arbitrage, not necessarily strong.

\subsection{Existence of strong short-term relative arbitrage} \label{SS:5.1}
 The following result greatly extends and simplifies  the results in Section~8 of \citet{FKK} and in Section~8 in \cite{FK_survey}. 

 \begin{thm}[One asset with sufficient variation]  
 \label{P:160715.1}
 In a market as in Section~\ref{S:2}, with relative weight processes $\mu_1 (\cdot), \cdots, \mu_d (\cdot)$, suppose there exists a constant $\eta>0$ such that $\langle \mu_1 \rangle (t) \geq \eta t$ holds on the stochastic interval   $\lc 0, \mathscr{D}^*\lc$  with
	\begin{equation*}
\mathscr D^{*} :=  \inf \left\{ t \ge 0 :\, \mu_1 (t) \leq \frac{\mu_1(0)}{2} \right\}.
\end{equation*}
Then, given any real number $T>0$ there exists a long-only trading strategy $\varphi (\cdot)$ which is strong relative arbitrage with respect to the market over the time horizon $[ 0, T]$. 
 \end{thm}
 
 Proposition~\ref{C:160610.1}  below is a direct consequence of Theorem~\ref{P:160715.1}; its proof requires, however, the technical observation made in Lemma~\ref{L:160711}. Proposition~\ref{C:160610.1} shows that, in the case $d=2,$ the condition in \eqref{eq: 4.3} yields,  for {\it every} given time horizon, the existence of a long-only trading strategy  which is  strong relative arbitrage with respect to the market over this  time horizon.  
 
 \begin{proof}[Proof of Theorem~\ref{P:160715.1}]
 	Let us fix a real number $T >0$ and consider the market with weights  $\nu(\cdot) := \mu(\cdot  \wedge \mathscr D^{*})$.
	It suffices to prove the existence of a long-only trading strategy $\varphi (\cdot)$ which is strong relative arbitrage with respect to the new market with weights $\nu (\cdot)$   over the time horizon $[ 0, T]$. For then the long-only trading strategy $\varphi (\cdot \wedge \mathscr{D}^*)$ is strong relative arbitrage with respect to original market with weights $\mu (\cdot)$   over the time horizon $[ 0, T]$.

	For some number $q > 1$ to be determined presently, we recall the  regular  function $\bm F$ of \eqref{eq:P}. Since $1/\nu_1(\cdot)$ is locally bounded, $\bm F$  generates multiplicatively, for the market with weight process $\nu(\cdot)$, the strategy $\psi^{\bm F}(\cdot)$ given by \eqref{eq:160717.1}, with $\mu(\cdot)$ replaced by $\nu(\cdot)$. We note   ${\bm \psi}^{\bm F}_i ( \cdot) \leq 1$   for  $i = 1, \cdots, d,$ 
and that $V^{{\bm \psi}^{\bm F}}(\cdot) = Z^{\bm F}(\cdot)$ is given as in \eqref{eq:160717.2}. We   introduce now the long-only trading strategy
\begin{align*}
	\varphi_i(\cdot) = 1 + \big(\nu_1(0)\big)^q - {\bm \psi}^{\bm F}_i(\cdot), \qquad i = 1, \cdots, d
\end{align*}
with associated wealth process 
\begin{align*}
	V^\varphi(\cdot) = 1 + \big(\nu_1(0)\big)^q- Z^{\bm F}(\cdot).
\end{align*}
In particular,  we note $V^\varphi(0) = 1$ and $V^\varphi(\cdot)  \geq 0$. On the event  $\{\mathscr D^{*}\leq T\}$ we have 
\begin{align*}
	V^\varphi(T) \geq  1 + \big(\nu_1(0)\big)^q  -\big(\nu_1(T)\big)^q=  1 + \big(\nu_1(0)\big)^q  - \left(\frac{\nu_1(0)}{2}\right)^q > 1;
\end{align*}
whereas, on the event $\{\mathscr D^{*} > T\}$ we have
\begin{align*}
		V^\varphi(\cdot) \geq 1 + \big(\nu_1(0)\big)^q -  \exp\left(-\frac{1}{2} q(q-1)  \langle \nu_1 \rangle (T)\right) \geq 1 + \big(\nu_1(0)\big)^q -  \left(\exp\left(-\frac{\eta}{2} (q-1)   T\right) \right)^q > 1,
\end{align*}
provided we choose $q = q(T)$ large enough such that 
\begin{align*}
	 \exp\left(-\frac{\eta}{2}  \big(q-1\big) T\right) <  \nu_1(0) = \mu_1(0).
\end{align*}
This shows $\P\big( V^\varphi (T) > 1 \big) =1$, as claimed. 
 \end{proof}

\begin{rem}[Dependence  of $\varphi(\cdot)$ on the length of the time horizon]
	The trading strategy $\varphi(\cdot)$ of Theorem~\ref{P:160715.1} is constructed in a very explicit and ``model-free" manner, but {\it does} depend on the length of the time horizon  over which it effects arbitrage with respect to the market.  
\end{rem}

The following two remarks recall two alternative ways to obtain the existence of strong relative arbitrage opportunities over arbitrary time horizons.

 \begin{rem}[Smallest asset with sufficient variation]  \label{R:160718.1}
In the spirit of Theorem~\ref{P:160715.1}, \cite{Banner_Fernholz} also prove the existence of a strong relative arbitrage over arbitrary time horizons. However, they do not assume that   one fixed asset  contributes to the overall market volatility -- but rather   that it is always the smallest stock that has sufficient variation.  The strategy they construct is again ``model-free'', but does depend  on the length of the time horizon.  
\end{rem}

\begin{rem}[Completeness and arbitrage imply strong arbitrage]  \label{R:completeness}
	If the underlying market model allows for a deflator (recall Definition~\ref{D:deflator}) and is complete (any contingent claim can be replicated),  then the existence of an arbitrage opportunity implies the existence of a strong one; see Theorem~8 in \cite{Ruf_ots}. 	However, this strong arbitrage usually will depend on the model and on the length of the time horizon.
\end{rem}

\subsubsection*{Diversity and strict nondegeneracy}

\begin{defn}[Diversity]
\label{def: 2.1}
We say that a market with relative weight processes  $\mu_1 (\cdot), \cdots,  \mu_d (\cdot)$ is {\it diverse}  if 
\begin{equation}
\label{eq: div}
\mathbb{P} \bigg( \sup_{t \in [0,\infty)}   \max_{1 \le i \le n} \mu_i (t) \le 1 -\delta \bigg) = 1  \qquad \text{holds for some   constant $\delta \in (0,1)$.}
\end{equation}
\end{defn}

Diversity posits that no company can come close to dominating the entire market capitalization. It is a typical characteristic of large, deep and liquid equity markets. 

Let us assume now that the continuous semimartingales $S_1 (\cdot), \cdots, S_d (\cdot)$ in \eqref{eq: 2.1a}, to wit, the capitalization processes of the various assets in the market, have covariations of the form 
\begin{equation}
\label{eq: 4.5}
\langle S_i, S_j \rangle (\cdot) = \int_0^{ \cdot} S_i (t) S_j (t) A_{i, j} (t) \, \dx t, \qquad  i,j =1, \cdots, d\end{equation}
for suitable progressively measurable processes $A_{i, j} (\cdot)$.

\begin{cor}[Diversity and strict nondegeneracy]
\label{cor: 5.11}
Suppose that a market as in \eqref{eq: 2.1a} is diverse, that \eqref{eq: 4.5} holds, and that the asset covariation matrix-valued process  $A  (\cdot) =( A_{i ,j} (\cdot))_{1 \le i, j \le d}$ satisfies the strict nondegeneracy condition 
\begin{equation}
\label{eq: 4.6}
 \xi' A (t,\omega)  \xi  \ge \lambda \, || \xi ||^2, \qquad \text{for all } \,~\xi \in \R^d, \,\,  (t,\omega) \in [ 0, \infty) \times \Omega
\end{equation} 
 for some   $\lambda > 0$. Then, given any real number $T>0$, there exists a long-only trading strategy $\varphi (\cdot)$ which is strong relative arbitrage with respect to the market over the time horizon $[ 0, T]$.
\end{cor}
\begin{proof}
With the help of (3.11) in \cite{FK_survey}, the conditions in \eqref{eq: div}  and \eqref{eq: 4.6}, namely, diversity and  strict nondegeneracy, lead  to the lower  bound
$$
\big \langle \mu_1 \big \rangle (T \wedge \mathscr D^{*} )  \ge   \lambda \int_0^{T \wedge \mathscr D^{*}  } \big( \mu_1 (t) \big)^2  \Big( 1 - \max_{1 \le i \le n} \mu_i (t) \Big)^2  \dx t   \ge    \eta  \big( T \wedge \mathscr D^{*} \big), \qquad T \geq 0 
$$
for $\eta := \lambda   \big( \delta  \mu_1 (0)  / 2\big)^2$. The claim is now a direct consequence of Theorem~\ref{P:160715.1}. 
\end{proof}

\subsection{Existence of short-term relative arbitrage, not necessarily strong}  
 \label{SS:5.2}

In this subsection we provide three more criteria that guarantee the existence of relative arbitrage with respect to the market over arbitrary time horizons. The first criterion is a condition on the time-homogeneity of the support of the market weight vector $\mu(\cdot)$. The second criterion uses failure of diversity. The third criterion is a condition on the nondegeneracy of the covariation process of $\mu(\cdot)$.

\subsubsection{Time-homogeneous support}  

Let us recall the condition in \eqref{eq:160131.3}. There, the threshold ${\bm G}(\mu(0)) $  ``is at its lowest",   when the initial market-weight configuration $\mu (0)$   is at a site   where  $\bm G$  attains, or is very close to, its smallest value on  ${\bm \Delta}^d$;  these are the most propitious  sites   from which relative arbitrage can be launched. 

 The   following result assumes that the essential infimum of the continuous semimartingale ${\bm G}(\mu (\cdot))$  remains constant over time. In particular, that sites in the state space which are  close to this infimum -- and thus ``propitious" for  launching relative arbitrage -- can be visited ``quickly" (i.e.,  over any given time horizon)  with positive probability. The result shows that relative arbitrage with respect to the market can then be realized  over {\it any} time horizon $[0,T]$ of arbitrary length $T>0$.

 \begin{thm}[A time-homogeneity condition on the support]   
 \label{T:5.1}
 Suppose    that  for  a given   generating  function $\bm G$ and appropriate real constants $ \eta >0$ and $ h \ge   0 $,   the condition  in \eqref{eq: 4.3} is satisfied, along with  the   lower bound 
\begin{equation}
 \label{10c}
 \P  \big(    {\bm G}(\mu ( t) )  \ge  h, \quad t \geq 0 \big) =1
 \end{equation}
and the     ``time homogeneous support" property 
\begin{equation}
 \label{10b}
\P \Big(    {\bm G}\big(\mu (  t)\big)\in \big[h, h+\varepsilon  \big), \,\,\,\text{for some } t \in [0,T] \Big)  > 0 , \qquad  \text{for all } T >0,\, \varepsilon > 0.
\end{equation}
Then    arbitrage relative to the market exists    over the time horizon $[0, T]$, for every real number $T >0$. 
\end{thm}
 The basic argument in the proof of Theorem~\ref{T:5.1}  is quite simple to describe: Given a time horizon $[0,T]$, the condition in \eqref{10b} declares that the  vector process $\mu (\cdot)$ of relative market weights will visit before time $T/ 2$, with positive probability,  sites which are ``very propitious'' for arbitrage relative to the market. The moment this happens, it makes good sense to abandon the market  in favour of the trading strategy $\bm \varphi^\star(\cdot)$ generated by the function $\bm G^\star = c (\bm G-h)$ in the manner of \eqref{eq:3.2}, for some appropriately chosen constant $c>0$.
 The challenge then is to  show that such a strategy does  not    underperform the market, and   has a positive probability of outperforming it strictly.

\begin{proof}[Proof of Theorem~\ref{T:5.1}]
For an arbitrary but fixed real number $T>0$ 
we introduce the regular function
\begin{align*}
	\bm G^\star := (\bm G-h)  \frac{3}{\eta T},
\end{align*}
and denote
\begin{align*}
 \Gamma^\star (\cdot) := \Gamma^{\bm G^\star} (\cdot)= \frac{3}{\eta T} \Gamma^{\bm G } (\cdot).
\end{align*}
We also  introduce the stopping time
$$
\tau := \inf \left\{  t \in \left[0, \frac{T}{2}\right] :\,  {\bm G}\big(\mu (t) \big) < h + \frac{\eta T}{3}    \right\}  
$$
with the usual convention that the infimum of the empty set is equal to infinity, and note  that \eqref{10b} implies
\begin{align}\label{eq:170617.4}
\P \left( \tau \le \frac{T}{2} \right)> 0.
\end{align}
We let $\bm \varphi^{\star} (\cdot):=  \bm \varphi^{\bm G^\star} (\cdot)$ denote the trading strategy generated by the function $\bm G^\star$ in the manner of \eqref{eq:3.2}, and introduce the trading strategy $\bm \varphi (\cdot)$ which follows the market portfolio up to the stopping time $  \tau$, then switches for the remainder of the time horizon $[0,T]$ to the trading strategy $\bm \varphi^{\star} (\cdot)$; namely,  
\begin{align*}
	 {\bm \varphi}_i(\cdot) := 1+ \big(\bm \varphi^\star_i(\cdot) - {\bm G}^\star(\mu(\tau)) - \Gamma^\star(\tau)\big) \1_{\lc \tau, \infty\lc}, \qquad i = 1, \cdots, d
\end{align*}
in the ``self-financed'' manner of \eqref{eq:160717.3}.
 According to \eqref{eq:170617.4}, this switching occurs with positive probability. As we saw in Remark~\ref{R:2.4},   the value that results form this concatenation is given by 
\begin{align*}
	 V^{{\bm \varphi}} (t) &= 1 +  \big( {\bm G}^\star(\mu(t))   + \Gamma^\star(t)    - {\bm G}^\star(\mu(\tau)) - \Gamma^\star(\tau) \big)  \1_{\lc \tau, \infty\lc}(t) \\
	 	&\geq \1_{\lc 0, \tau\lc}(t) +  \frac{3}{\eta T} \left(\Gamma^{\bm G}(t)  - \Gamma^{\bm G}(\tau) \right) \1_{\lc \tau, \infty\lc}(t) \\
		&\geq  \1_{\lc 0, \tau\lc}(t) +  \frac{3 (t - \tau)}{T}  \1_{\lc \tau, \infty\lc}(t), \qquad t \geq 0.
\end{align*}
Here we have used the comparisons ${\bm G}^\star(\mu(\cdot)) \geq 0$ and ${\bm G}^\star(\mu(\tau)) \leq 1$ in the first inequality, and \eqref{eq: 4.3} in the second inequality.
 
Now it clear from this last display  that $ V^{\bm \varphi} (\cdot) \geq 0$ holds, and that $V^{\bm \varphi} (T) \ge 3/2$ holds on  the event  $\{   \tau \le T / 2 \}$; it is also clear  that $V^{\bm \varphi} (T) =1$ holds on   $\{   \tau > T / 2 \} = \{   \tau = \infty \}$. Since   $\{   \tau \le T / 2 \}$ has positive probability on account of  \eqref{eq:170617.4}, it follows that the trading strategy $\bm \varphi (\cdot)$ is relative arbitrage with respect to the market over the time horizon $[0,T]$.
\end{proof}
 
\begin{rem}[On the type of arbitrage]
There is nothing in the above argument  to suggest that the probability in \eqref{eq:170617.4}, which is argued there to be positive, is actually equal  to 1. Thus, the relative arbitrage constructed in Theorem~\ref{T:5.1} need not be strong. 
It should also be noted that the trading strategy which implements this arbitrage depends  on the length $T$ of the time horizon $[0,T]$ -- in marked contrast to the strategy of Theorem~\ref{T:4.3}, which does not, as long as    $T \geq T_*$.  
\end{rem}

\subsubsection{Failure of diversity}  

Theorem~\ref{T:5.1}    has the following corollary.  Taken together, Corollaries~\ref{cor: 5.11} and \ref{cor: 5.3} illustrate that   both diversity, {\it and} its failure, can lead to arbitrage over arbitrary time horizons -- under appropriate additional conditions in each case. For the statement, we  let  $\,\mathfrak{e}_1, \cdots, \mathfrak{e}_d\,$ denote the extremal points (unit vectors) of the lateral face $\bm \Delta^d$ of the unit simplex.

\begin{cor}[Failure of diversity]
\label{cor: 5.3}
Suppose that diversity     fails for a market with   relative weights $\mu (\cdot)$, in the sense that 
\begin{equation*}
\P\bigg( \sup_{t \in [0,T)}   \max_{1 \le i \le d} \mu_i (t) > 1 -\delta \bigg) >0 \qquad \hbox{holds for every } ~(T, \delta) \in (0, \infty) \times (0,1).
\end{equation*}
 Suppose also that, for some generating  function $\bm G$ which satisfies
 \begin{equation*}
 \bm G(\mathfrak e_i) =  \min_{x \in \bm \Delta^d} \bm G(x), \qquad \text{for each }~~i=1, \cdots, d, 
\end{equation*} 
the condition in \eqref{eq: 4.3} holds for some real constant  $\eta >0$. Relative arbitrage with respect to the market   exists  then over the  timehorizon $[0,T]$,   for every real number     $T>0$.
\end{cor}

\subsubsection{Strict nondegeneracy}  

Theorem~\ref{T:5.1} has     an important consequence, Theorem~\ref{T:5.2} below; this  establishes the existence of relative arbitrage with respect to the market, under the ``sufficient intrinsic volatility" condition of \eqref{eq: 4.3} and under additional nondegeneracy conditions. 

In order to prepare the ground for this result, let us recall the trace process $\Gamma^{\bm Q} (\cdot  )$ of \eqref{eq:GQ}.    From Proposition~II.2.9 of \cite{JacodS} we have the representation
	\begin{align}  
	\label{eq:160524.1}
		\big \langle \mu_i, \mu_j \big \rangle(\cdot) = \int_0^\cdot \alpha_{i,j}(t) \, \dx \Gamma^{\bm Q}(t), \qquad 1 \leq i,j \leq d
	\end{align}
for some symmetric and nonnegative-definite matrix-valued process $\alpha  (\cdot) = ( \alpha_{i,j} (\cdot) )_{i,j = 1, \cdots, d}$, whose entries are progressively measurable and satisfy $\sum_{j=1}^d \alpha_{i,j} (\cdot) \equiv 0$ for every $i=1, \cdots, d$.  Furthermore, thanks to the Kunita-Watanabe inequality (Proposition~3.2.14 in \cite{KS1}), the process $\alpha_{i,j} (\cdot)$ takes values in $[-1, 1]$ for every $i,j=1, \cdots, d$.

We also consider  the sequence of stopping times 
\begin{equation}
\label{eq:160726.1}
\mathscr D^n :=  \inf \left\{ t \ge 0 :\, \min_{1 \leq i \leq d}  \mu_i (t) < \frac{1}{n} \right\}, \qquad n \in \N.
\end{equation}

\begin{thm}[A strict nondegeneracy condition]  \label{T:5.2}
Suppose that there exists a deflator for the process $\mu (\cdot) = \big( \mu_1 (\cdot), \cdots,   \mu_d (\cdot) \big)'$ of relative market weights, as well as a regular function $\bm G$ which satisfies \eqref{eq: 4.3} for some real constant $ \eta >0$. 
Moreover, assume that the $d-1$ largest eigenvalues of the matrix-valued process $\alpha(\cdot)$ in \eqref{eq:160524.1} are bounded away from zero  on $\lc 0, \mathscr D^n\lc$  uniformly in $(t, \omega)$, for each $n \in \N$, in the notation of \eqref{eq:160726.1}.
Then  relative arbitrage with respect to the market  exists over   $[0,T]$, for every real number $T>0$. 
\end{thm}
The proof of Theorem~\ref{T:5.2} is   at the end of this subsection. Proposition~\ref{P:160712.2} below shows that, in this theorem, it is not sufficient that the $d-1$ largest eigenvalues of the matrix-valued process $\alpha(\cdot)$ be strictly positive. If they are not additionally 
bounded away from zero, an example can be constructed in which relative arbitrage with respect to the market  does not exist over the time horizon $[0,T]$ for some real number $T>0$.

It is   important to stress that Theorem~\ref{T:5.2} establishes only the {\it existence} of a  trading strategy, which effects the claimed relative arbitrage. Moreover, unlike the trading strategy of Theorem~\ref{T:4.3}  which is  strong arbitrage, explicit, model-free, and independent of the time horizon, the trading strategy whose existence is claimed in Theorem~\ref{T:5.2} may be none of these things.

\begin{rem}[It\^o-process covariation structure] \label{R:Ito}
	If 
	$\langle \mu_i, \mu_j \rangle(\cdot) =  \int_0^\cdot\beta_{i,j}(t) \dx t$  holds for all $\i,j = 1, \cdots, d$,  then 
	\begin{align*}
		\Gamma^{\bm Q}(\cdot) &=  \sum_{j = 1}^d \int_0^\cdot \beta_{j,j}(t) \dx t;\\
		 \alpha_{i,j}(\cdot) \1_{\{\sum_{k = 1}^d \beta_{k,k}(\cdot) > 0\}} &= \frac{\beta_{i,j}(\cdot)}{ \sum_{k = 1}^d \beta_{k,k}(\cdot)} \1_{\{ \sum_{k = 1}^d \beta_{k,k}(\cdot) > 0\}}, \qquad i,j = 1, \cdots, d.
	\end{align*}
	Hence, in this case, a sufficient (though not necessary) condition for the nondegeneracy assumption in Theorem~\ref{T:5.2} to hold,  is that  the $d-1$ largest eigenvalues of the matrix-valued process $\beta(\cdot)$ be bounded away from zero and from infinity on $\lc 0, \mathscr D^n\lc$,  for each $n \in \N$.
\end{rem}

The proof of Theorem~\ref{T:5.2} uses the following lemma.

\begin{lem}[Sum of quadratic variations bounded from below]  
\label{L:160711}
Assume that there exist  a generating  function $\bm G$ and a constant  $ \eta >0$ such that    \eqref{eq: 4.3} is satisfied.  Then, for each $n \in \N$, there exists a real constant $C=C(n, \eta, d, \bm G)$ such that the mapping  $t \mapsto  \Gamma^{\bm Q}(t) - C t$ is nondecreasing on $\lc 0, \mathscr D^n\lc$. 
\end{lem}
\begin{proof}
Let us fix $n \in \N$. Thanks to \eqref{eq: Gamma} we get
	\begin{equation*}
  \Gamma^{\bm G} (\cdot)= -
\frac{1}{2}  \sum_{i=1}^d \sum_{j=1}^d \int_0^{ \cdot}    D^2_{i,j} {\bm G}\big(\mu (t)  \big)  \alpha_{i,j} (t) \,\dx \Gamma^{\bm Q}(t) \qquad \text{  on $\lc 0, \mathscr D^n\lc$.}
     \end{equation*}
  Next, we observe that $\sum_{i=1}^d \sum_{j=1}^d \big|D^2_{i,j} {\bm G}\big(\mu (\cdot)  \big) \big| $ is bounded from above  by a real constant $K_n>0$ on the stochastic interval $\lc 0, \mathscr D^n\lc$. In light of the inequality $|\alpha_{i,j}(\cdot)| \leq 1$ for all $1 \le i,j  \le d$, we have that the difference 
     \begin{equation}
     \label{eq: Gamma_QG}
    K_n    \Gamma^{\bm Q}(\cdot) - 2 \Gamma^{\bm G} (\cdot) \qquad \text{is nondecreasing on ~~$\lc 0, \mathscr D^n\lc$}.
      \end{equation}
Hence,  \eqref{eq: 4.3} yields the statement with $C = 2 \eta/ K_n$.	
\end{proof}

\begin{prop}[The case of two assets]  \label{C:160610.1} 	Assume that $d=2$ and that there exist a regular   function $\bm G$ and a real constant $ \eta >0$ such that    \eqref{eq: 4.3} is satisfied. 
	Then  strong arbitrage relative to the market can be realized by a long-only trading strategy over the time horizon $[0, T]$,  for any given real number $T>0$.
\end{prop}

\begin{proof}
This follows directly from Lemma~\ref{L:160711} and  Theorem~\ref{P:160715.1}, as in this case $\Gamma^{\bm Q}(\cdot)= 2  \langle \mu_1 \rangle (\cdot)$.
\end{proof}

Alternatively, a weaker formulation of Proposition~\ref{C:160610.1},  which guarantees the existence of relative arbitrage over any given time horizon, but not the fact that this  relative arbitrage is strong,  can  also be proved via Theorem~\ref{T:5.2}. To apply this result, it suffices  to check that the largest eigenvalue of $\alpha(\cdot)$ equals $1$. However, this is easy to see here: we have $\mu_2(\cdot) = 1 - \mu_1(\cdot)$; hence $\alpha_{1,1}(\cdot) = \alpha_{2,2}(\cdot) = 1/2$ and  $\alpha_{1,2}(\cdot) = \alpha_{2,1}(\cdot) = -1/2$, so the eigenvalues of the matrix $\alpha(\cdot)$ are then indeed $0$ and $1$.

\begin{proof}[Proof of Theorem~\ref{T:5.2}]
First,   we may assume without loss of generality that $\bm G$ is nonnegative.
We shall argue by contradiction, assuming  that for some real number $T_* >0$ no relative arbitrage  is possible with respect to the market on the time horizon $[0, T_*]$.  Remark~\ref{R:4.2} gives then the existence of an  equivalent probability measure $\Qu_* \sim \P$ on on $\sigalg F(T_*)$,  under which the relative weights $\mu_1 (\cdot \wedge T_*), \cdots, \mu_d (\cdot \wedge T_*)$ are   martingales. 

We shall show next that this leads to the property \eqref{10b} with $h = \min_{x \in \bm \Delta^d}\bm G(x)$
and hence, on the strength of Theorem~\ref{T:5.1}, to the desired contradiction. 
In order to make headway with this approach we fix $\varepsilon > 0$ and $T \in (0,T_*]$, define  
\[
	\mathcal U := G^{-1}\big([h, h+\varepsilon)\big)  \cap (0,1)^d \subset \bm \Delta^d_+,
\]
 choose  a point $x \in \mathcal U$, and fix an integer $N \in \N$ large enough so that 
\begin{align*}
	\min_{1 \leq i \leq d} x_i > \frac{2}{N}; \qquad \min_{1 \leq i \leq d} \mu_i(0) > \frac{2}{N}.\end{align*}
We recall the constant $C = C(N, \eta, d, \bm G)$ from Lemma~\ref{L:160711} and define the stopping time
\begin{equation} \label{eq: 5.13}
\bm \varrho :=  \inf \left\{ t \ge 0 : \Gamma^{\bm Q}(t) > C T\right\}.
\end{equation}
For future reference, we  note that Lemma~\ref{L:160711} yields the set inclusion
\begin{align}\label{eq:160712.7}
\{\mathscr D^N \geq T\} \subset \{\bm \varrho \leq T\}.
\end{align}

Now,  in order to obtain  \eqref{10b} , it suffices to show that the stopped process
\begin{equation}
\label{eq: 5.14}
\nu (\cdot) := \mu \big( \cdot  \wedge  {\bm \varrho} \big)
\end{equation}
satisfies   $\Q_* \big( \nu (T) \in {\cal U} \big) >0$. This, in turn, will follow as soon as we have shown 
  \begin{align}  \label{eq:160530.1}
 	\Q_*\left(\sum_{i = 1}^d (\nu_i( T) - x_i)^2 < \delta \right) > 0.
 \end{align}
 Here $\delta \in (0,1/N)$ is sufficiently small so that, for all $y \in \bm \Delta^d$, we have
 \begin{align*}
 	\sum_{i =1}^d (y_i - x_i)^2 < \delta\qquad &\text{implies} \qquad y \in \mathcal U ~~~~\text{ and } ~~\min_{1 \leq i \leq d} y_i > \frac{1}{N};\\
 	\sum_{i =1}^d (y_i - \mu_i(0))^2 < \delta \qquad &\text{implies} \qquad \min_{1 \leq i \leq d} y_i > \frac{1}{N}.
 \end{align*}
 Clearly, on account of \eqref{eq: 5.13}, we have for the stopped process $\nu (\cdot)$ in \eqref{eq: 5.14} the upper bound
\begin{equation}
\label{eq: 5.16}
\sum_{i=1}^d \langle \nu_i \rangle (\cdot) = \sum_{i=1}^d  \langle \mu_i \rangle (\cdot  \wedge {\bm \varrho}) = \Gamma^{\bm Q}  (\cdot \wedge {\bm \varrho})  \le CT.
\end{equation}

In order to establish \eqref{eq:160530.1}, we modify the arguments in \citet{SV_support} and  \citet{Stroock:1971}. We fix the vector
\begin{align} \label{eq: 5.17}
	\zeta :=  \frac{x-\mu(0)} {C T}
\end{align}
and let $\alpha^\flat (\cdot)$ denote the {Moore-Penrose} pseudo-inverse of the matrix $\alpha  (\cdot)$ in \eqref{eq:160524.1}.
Next, note that the vector $\zeta$ is in the range of $\alpha(\cdot)$ on $\lc 0, \mathscr D^N\lc$ since the matrix $\alpha(\cdot)$ has rank $d-1$ and satisfies $\alpha(\cdot) \mathfrak e = 0$, where $\mathfrak e = (1, \cdots, 1)'$. 
Thus, we have $\alpha(\cdot) \alpha^\flat(\cdot) \zeta = \zeta$  on $\lc 0, \mathscr D^N\lc$.
We now introduce the continuous $\Q_*$--local martingale
\begin{equation*}
M(\cdot) := \int_0^{\cdot \wedge \mathscr D^N} \big \langle \alpha^\flat (t) \zeta, \mathrm{d}\nu(t) \big \rangle  = \int_0^{\cdot \wedge \mathscr D^N \wedge \bm \varrho} \big \langle \alpha^\flat (t) \zeta, \mathrm{d}\mu(t) \big \rangle.
\end{equation*}
The quadratic variation of this local martingale is dominated by a real constant: namely, 
 \begin{align*}
 	\langle M\rangle (\cdot) &=  \int_0^{\cdot \wedge \mathscr D^N \wedge {\bm \varrho}}   \big( \zeta'   \alpha^\flat ( t) \big)  \alpha  ( t)  \big(     \alpha^\flat ( t) \zeta \big)    \dx \Gamma^{\bm Q} (t)
	= \int_0^{\cdot \wedge \mathscr D^N \wedge {\bm \varrho}}    \zeta'   \alpha^\flat ( t) \zeta    \dx \Gamma^{\bm Q} (t)\\
 &\leq \frac{\zeta' \zeta }{c_N}  \Gamma^{\bm Q} (\cdot \wedge \mathscr D^N \wedge {\bm \varrho})  \leq \frac{1}{CT c_N} \sum_{i = 1}^d (x_i - \mu_i(0))^2 \leq \frac{1} {CT  c_N},
 \end{align*} 
 on account of \eqref{eq: 5.16} and \eqref{eq: 5.17}, where the real constant $c_N>0$ stands for a lower bound on the smallest positive eigenvalue of $\alpha (\cdot)$ on the stochastic interval $\lc 0, \mathscr{D}^N\lc$. Likewise, we have
 \begin{align}
 	\langle \nu_i, M\rangle(\cdot) &= \sum_{j=1}^d \sum_{k=1}^d \int_0^{\cdot \wedge \mathscr D^N  \wedge {\bm \varrho} }  
	   \zeta_k \alpha^\flat_{k,j} (t) \alpha_{i,j} (t)  \dx \Gamma^{\bm Q} (t) 
	 = \zeta_i   \Gamma^{\bm Q} \big(\cdot \wedge \mathscr{D}^K \wedge {\bm \varrho}\big)  \nonumber\\
	 &= \frac{x_i - \mu_i(0)}{CT}  \Gamma^{\bm Q} \big(\cdot \wedge \mathscr{D}^K \wedge {\bm \varrho}\big) , \qquad i = 1, \cdots, d.   \label{eq:160712.6}
 \end{align}

On account of Novikov's theorem (Proposition~3.5.12  in \cite{KS1}), the stochastic exponential  $ \mathcal E(M(\cdot))$
is a uniformly integrable $\Qu_*$--martingale. Thus, this exponential martingale generates a new probability measure $\Qu$ on $\sigalg F(T_*)$, which is equivalent to $\Qu_*$. According to the {van\,Schuppen-Wong} extension of the {Girsanov} theorem (ibid.,  Exercise~3.5.20), we have then the  decomposition 
\begin{align*}
	\nu_i(\cdot) = \mu_i(0) + \langle \nu_i, M\rangle(\cdot)  + X_i(\cdot), \qquad i=1, \cdots, d,
\end{align*}
where $X_i(\cdot)$ is a $\Qu$--local martingale with $X_i(0) = 0$.

We now consider the event
\begin{align*}
	A := \left\{\max_{t \leq T} \sum_{ i = 1}^d X_i^2(t) < \delta\right\} .
\end{align*} 
Thanks to \eqref{eq:160712.6}, any vector with components $ \mu_i(0) + \langle\mu_i, M\rangle(\cdot)$ for each $i = 1, \cdots, d$ is a convex combination of $\mu(0)$ and $x$; this leads to $A \subset \{\mathscr D^N \geq T\}$.  Now, in conjunction with \eqref{eq:160712.7} and \eqref{eq:160712.6}, this set inclusion implies that
$$\langle\nu_i, M\rangle(T) = x_i - \mu_i(0) \qquad \text{on $A$}, \qquad i = 1, \cdots, d,$$ 
and therefore 
\begin{align*}
	\sum_{i = 1}^{d} (\nu_i(T) - x_i)^2 = \sum_{i = 1}^{d} (\nu_i(T) - \mu_i(0) - \langle\nu_i, M\rangle(T))^2 = \sum_{i = 1}^{d} X_i^2(T) < \delta \qquad \text{on $A$}.
\end{align*}
Consequently, the claim \eqref{eq:160530.1} will follow  from the equivalence $\Q_* \sim \Q$, as soon as we have established that  $\Q(A) > 0$.

In order to argue this positivity, we start by introducing the processes 
\begin{align*}
	R(\cdot) := \sum_{i = 1}^d X_i^2(\cdot) \qquad \text{and}  \qquad Y(\cdot) := \int_0^\cdot
		 \1_{\{R(t)\geq \delta / 4\}}   \dx R(\cdot) \geq R(t) - \frac{\delta}{4}
\end{align*}
and noting 
\begin{align}   \label{eq:160910.2}
	N(\cdot) := Y(\cdot) - \sum_{i = 1}^d \int_0^{\cdot \wedge {\bm \varrho}}  \1_{\{R(t)\geq \delta/ 4\}}  \alpha_{i,i}(t)  \dx \Gamma^{\bm Q}(t)
	= 2 \int_0^\cdot    \1_{\{R(t)\geq \delta / 4\}}  \langle X(t), \dx X(t) \rangle.
\end{align}
Hence, it suffices to argue that
\begin{align} \label{eq:160910.1}
	 \Qu\left(\max_{t \leq T} Y(t) < \frac{\delta}{2} \right) > 0.
\end{align}

To this end, define  the $\Qu$--local martingale
	\begin{align*}
		\widehat M(\cdot) =- \frac{1}{2} \int_0^\cdot  \1_{\{R(t)\geq \delta / 4\}}   
		 \frac{\sum_{i = 1}^d \alpha_{i,i}(t)}{X'(t) \alpha(t) X(t)} \big\langle X(t) , \dx X(t) \big \rangle,
	\end{align*}
	whose quadratic variation is dominated again by a real constant: namely,
	\begin{align*}
		\langle \widehat M\rangle(\cdot) &= \frac{1}{4} \int_0^{\cdot \wedge {\bm \varrho}}  \1_{\{R(t)\geq \delta / 4\}}   
		 \frac{\big(\sum_{i = 1}^d \alpha_{i,i}(t)\big)^2}{X' (t) \alpha(t) X(t)} \,\dx \Gamma^{\bm Q}(t)  \leq  \frac{4}{c_N \delta^2} \int_0^\cdot \Big(\sum_{i = 1}^d \alpha_{i,i}(t)\Big)^2\, \dx \Gamma^{\bm Q}(t)\\
		 & = \frac{4 d^2 \Gamma^{\bm Q}(\cdot \wedge {\bm \varrho})}{c_N \delta^2} \leq \frac{4 d^2 C T}{c_N \delta^2}
	\end{align*}
	on account of \eqref{eq: 5.16}. Here the real constant $c_N>0$ stands again for a lower bound on the smallest positive eigenvalue of $\alpha (\cdot)$ on the stochastic interval $\lc 0, \mathscr{D}^N\lc$. 
Recalling the $\Q$--local martingale  $N(\cdot)$ in \eqref{eq:160910.2}, we obtain  
 \begin{align*}
	\big \langle N, \widehat M  \,\big \rangle(\cdot) = -   \sum_{i = 1}^d \int_0^{\cdot \wedge {\bm \varrho}}  \1_{\{R(t)\geq \delta / 4\}}  \alpha_{i,i}(t) \dx \Gamma^{\bm Q}(t).
\end{align*}
 
Another application of Novikov's theorem  yields that the stochastic exponential $\mathcal E(\widehat M(\cdot))$  is a uniformly integrable $\Qu$--martingale and generates a probability measure $\widehat \Qu$ on $\sigalg F(T_*)$, which is equivalent to $\Q$.
Under this new probability measure $\widehat \Q$, the process $Y(\cdot) = N(\cdot) - \langle N, \widehat M\rangle(\cdot)$ is a local martingale.
 However, this continuous $\widehat\Q$--local martingale   $Y(\cdot)$ is bounded from below by $- \delta /4$ and satisfies $Y(0)=0$;  thus the property 
\begin{align*}
	\widehat \Qu\left(\max_{t \leq T} Y(t) < \frac{\delta}{2} \right) > 0 
\end{align*}
holds, and \eqref{eq:160910.1} follows on the strength of the equivalence $\Q \sim \widehat \Q$. This concludes the proof.
 \end{proof}
 
 In the absence of strict nondegeneracy conditions as in Theorem~\ref{T:5.2},
 the controllability approach of \cite{Kunita:1974, Kunita:1978} yields conditions which guarantee that the assumptions of Theorem~\ref{T:5.1} hold when the vector market weight process $\mu (\cdot) = (\mu_1(\cdot) , \cdots, \mu_d (\cdot))'$ is an It\^o diffusion.  In the same spirit,  suitable H\"ormander-type hypoellipticity conditions on the covariations of the components of this diffusion, along with additional technical conditions on the drifts,   yield good ``tube estimates" which then again avoid the need to impose  strict nondegeneracy conditions; see  \cite{Bally:II:2016}, and the literature cited there.

\section{Lack of short-term relative arbitrage opportunities} 
 \label{S:6}

In Remark~\ref{R:old} we raised  the question, whether the condition of \eqref{eq: 4.3} yields the existence of relative arbitrage  over sufficiently short time horizons. In Section~\ref{S:5} we saw that, under appropriate {\it additional conditions} on the covariation  structure of the market weights, the answer to this question is affirmative. 
In general, however, the answer to the question of Remark~\ref{R:old} is negative, as we shall see in the present section. Specific counterexamples of market models will be constructed in a systematic way, to illustrate that arbitrage opportunities over arbitrarily short time horizons do not necessarily exist in   models which satisfy \eqref{eq: 4.3}.
 
In Subsections~\ref{SS:6.1}, \ref{SS:6.2}, and \ref{SS:6.3} we shall focus on the the quadratic function $  \bm Q$ of  \eqref{eq:Q}. More precisely, we   construct there variations of market models $\mu(\cdot)$ that satisfy \eqref{eq: 4.3} with $\bm G = \bm Q$, but do not  admit  relative arbitrage over any time horizon.  We recall that $ 2\Gamma^{\bm H}(\cdot) - \Gamma^{\bm Q} (\cdot) $ is nondecreasing for the cumulative excess growth $\Gamma^{\bm H}(\cdot)$ of \eqref{eq: 3.1}; thus, if \eqref{eq: 4.3} is satisfied   by the quadratic function $  \bm Q$, it is automatically also satisfied by the entropy function $
\bm H$   of \eqref{eq: 3.13}.  This then also yields a negative answer to the question posed in Remark~\ref{R:old}. 
In Subsection~\ref{SS:6.4}, market weight models $\mu(\cdot)$ are constructed, such that $\bm G(\mu(\cdot))$ moves along the level sets of a general Lyapunov function $\bm G$ at unit speed, namely, with $\bm G (\mu (t)) = \bm G (\mu (0)) - t$ and $\Gamma^{\bm G} (t)=t$, but which do not admit relative arbitrage over any time horizon. 

\begin{rem}[Some simplifications for notational convenience] \label{R:6.1}
Throughout this Section we shall make certain assumptions, mostly for notational convenience.
\begin{itemize}
\item We shall focus on  the case $d=3$.  Indeed,  Proposition~\ref{C:160610.1} shows that it would be impossible to find a counterexample to the question of Remark~\ref{R:old} when $d=2$.  The counterexamples below can be generalized to more than three assets,  but at the cost of additional notation  and without any major additional insights.
\item We shall construct   market models that satisfy, for a certain Lyapunov function $\bm G$ the condition
\begin{equation}
\label{eq: 4.3mod}
\P \left( \text{the mapping} ~ [0,T] \ni t \mapsto \Gamma^{\bm G}(t)- \eta t \,\,\,\text{is nondecreasing} \right)= 1, \quad \text{for some  $\eta >0$ and $T>0$}.
\end{equation}
We may do this without loss of generality.  Indeed, let us assume we have a market model $\mu(\cdot)$ that satisfies \eqref{eq: 4.3mod} and does not allow relative arbitrage over any time horizon.  By appropriately adjusting the dynamics of $\mu(\cdot)$, say after time $T/2$, it is then always possible to construct a market model $\widehat \mu(\cdot)$  that satisfies \eqref{eq: 4.3}, and also does not allow for arbitrage over short time horizons (as it displays the same dynamics up to time $T/2$).
\end{itemize}
\end{rem}

\subsection{A first step for the quadratic generating function} 
\label{SS:6.1} 

Here is a first result on absence of relative arbitrage under the condition of \eqref{eq: 4.3mod}.

\begin{prop}[Counterexample with Lipschitz-continuous dispersion matrix]
\label{P:160712.1}
Assume that the filtered probability space $(\Omega, \sigalg{F},   \P), $ $ \filt{F} = (\sigalg{F} (t) )_{t \geq 0 }$ supports a Brownian motion $W(\cdot)$.  Then there exists an  It\^o diffusion  $\mu (\cdot)= ( \mu_1 (\cdot  ), \mu_2 (\cdot  ), \mu_3 (\cdot  ))'$  with values in ${\bm \Delta}^3$,     a time-homogeneous and Lipschitz-continuous dispersion matrix in  ${\bm \Delta}^3_+$,  and the following properties: 
\begin{enumerate}[label={\rm(\roman{enumi})}, ref={\rm(\roman{enumi})}]
\item  No relative  arbitrage exists with respect to the market with  relative weight process $\mu (\cdot)$,    over any time horizon $[0,T]$ with  $T > 0$. 
\item The condition of \eqref{eq: 4.3mod} is satisfied by the  
quadratic   $\bm G = \bm Q$ of \eqref{eq:Q} 
with $\eta = 2/3 - \bm Q(\mu(0))$ and with $T =     T^* (\mu (0)) $   a strictly positive real number, provided that $\bm Q(\mu(0)) \in \big( 1/2, 2/3\big)$.  
\end{enumerate}
\end{prop} 

We provide the proof of Proposition \ref{P:160712.1}  at the end of this subsection. The following system of stochastic differential equations will play a fundamental role when deriving the dynamics of the relative market weight process $\mu(\cdot)$ in this proof: 
 \begin{align}
	\dx v_1(t) &= \frac{1}{\sqrt 3} \big (v_2(t) - v_3(t)\big) \,\dx W(t),\qquad t \geq 0,   \label{eq:160512.4}\\
	\dx v_2(t) &= \frac{1}{\sqrt 3} \big(v_3(t) - v_1(t)\big) \, \dx W(t),\qquad t \geq 0,  \label{eq:160512.5}\\
	\dx v_3(t) &= \frac{1}{\sqrt 3}  \big(v_1(t) - v_2(t)\big)\, \dx W(t), \qquad t \geq 0, \label{eq:160512.6}
\end{align}
where  $W(\cdot)$ denotes a Brownian motion.  The Lipschitz continuity of the coefficients guarantees that this system has a pathwise unique strong solution for any initial point $v(0) \in \R^3$.  If, moreover, $v(0)$ is  point in the hyperplane $ \mathbb{H}^3$ of \eqref{eq:160510.1}, then we also have $v(t) \in \mathbb{H}^3$, for each $t \geq 0$.

\begin{rem}[Explicit solution]  
\label{R:160712.1}
Let us provide an explicit solution $v(\cdot)$ of the system in  \eqref{eq:160512.4}--\eqref{eq:160512.6} provided that $v(0) \in \mathbb{H}^3$. If $v(0) = (1/3,1/3,1/3)',$ then we have also   $v(t) = (1/3,1/3,1/3)'$ for all $t \geq 0$. More generally, some determined but fairly basic stochastic calculus shows that  the solution of \eqref{eq:160512.4}--\eqref{eq:160512.6} is given by 
\begin{align*}
	v_1(t)  =   \frac{1}{3}  +    \frac{e^{t/2}}{3} & \Big[ 2 v_1(0) \cos(W(t)) +
	v_2(0)  \left(-\cos(W(t)) + \sqrt{3} \sin(W(t)) \right)  \\
	&~~~~~~~+ v_3(0)
	 \left(-\cos(W(t)) - \sqrt{3} \sin(W(t)) \right)   \Big] ;\\
	v_2(t)  =  \frac{1}{3}  +     \frac{e^{t/2}}{3} & \Big[ v_1(0) 
		\left(-\cos(W(t)) - \sqrt{3} \sin(W(t)) \right) +
		2 v_2(0)  \cos(W(t))   \\
	&~~~~~~~+  v_3(0)
		 \left(-\cos(W(t)) + \sqrt{3} \sin(W(t)) \right)   \Big] ;\\
	v_3(t) =   \frac{1}{3} +    \frac{e^{t/2}}{3} & \Big[ v_1(0) 
		\left(-\cos(W(t)) + \sqrt{3} \sin(W(t)) \right) +
		v_2(0)   \left(-\cos(W(t)) - \sqrt{3} \sin(W(t)) \right)   \\
	&~~~~~~~+  2 v_3(0)
		 \cos \big(W(t)\big)  \Big]. 
\end{align*}
\end{rem}

\begin{rem}[Representation in a special case]  
\label{R:160725.1}
With the initial condition
	\begin{align*}
		v_i(0) = \frac{1}{3} + \delta \cos \left(2 \pi \left( u   + \frac{i-1}{3}  \right) \right), \qquad i  = 1,2,3 
	\end{align*}
for some  $\delta \in [0,1/3]$ and $u \in \R$,  we find another useful representation of the solution in Remark \ref{R:160712.1}.
 Indeed, a computation shows 
\begin{align}  \label{eq:160514.1}
	\sum_{i=1}^3 v_i(0) = 1 + \delta \left(\cos(2 \pi u) + \cos \left(2 \pi  u   + \frac{2 \pi}{3}  \right) +  \cos \left(2 \pi  u   + \frac{4 \pi}{3}  \right)  \right)  = 1,
\end{align}
hence $v(0) \in \mathbb{H}^3$.  We now claim that 
		\begin{align} \label{eq:160712.3}
		v_i(t) = \frac{1}{3} + \delta  e^{t/2} \cos \left(W(t) + 2 \pi \left( u   + \frac{i-1}{3}  \right) \right), \qquad i  = 1,2,3,\,\, t \geq 0
	\end{align}
		solves the system \eqref{eq:160512.4}--\eqref{eq:160512.6}.  To this end, note that
 It\^o's formula yields the dynamics
\begin{align*}
	\dx v_i(t) = - \delta  e^{t/2} \sin \left(W(t) + 2 \pi \left( u   + \frac{i-1}{3}  \right) \right) \dx W(t), \qquad i  = 1,2,3,\,\, t \geq 0.
\end{align*}
 Moreover, since $\sin(\pi/3) = \sqrt{3}/2$,  it suffices to argue that
\begin{align*}
	 2 \sin\left(\frac{\pi}{3}\right)  \sin \left(x   + \frac{2\pi (i-1)}{3}   \right)
	=  \cos \left(x   + \frac{2\pi (i+1)}{3}   \right) - \cos \left(x   + \frac{2\pi i}{3}   \right), \qquad i = 1,2,3,\,\,x \in \R.
\end{align*}
 This is a basic trigonometric identity, from which the claim follows.     
\end{rem}

To study the dynamics of  the solution $v(\cdot)$ for the system in  \eqref{eq:160512.4}--\eqref{eq:160512.6} further, we introduce the   function:
\begin{align} \label{160512.4}
	{\bm r}: \R^3 \rightarrow [0,\infty), \qquad x \mapsto {\bm r} (x) := {\frac{1}{3}  \bigg( \big(x_1 - x_2\big)^2 +
	\big(x_1 - x_3\big)^2 +  
	\big(x_2 - x_3\big)^2}  \bigg). 
\end{align}
The following result shows that $\sqrt{{\bm r} (x)}$ is the distance from the ``node" $( 1/3 , 1/3, 1/3)'$ on the lateral face of the unit simplex.  

\begin{lem}[Another representation for ${\bm r}$]  
\label{L:160509.1}
	We have the representation
	\begin{align}   
	\label{eq:160718.1}
		{\bm r}(x) =   \sum_{i = 1}^3 \left(x_i-\frac{1}{3}\right)^2  = \sum_{i =1}^3 x_i^2 - \frac{1}{3}, \qquad x \in  \mathbb{H}^3
	\end{align}
in the notation of \eqref{eq:160510.1}.	Moreover, if $v(\cdot)$ denotes a solution to  \eqref{eq:160512.4}--\eqref{eq:160512.6} with $v(0) \in \mathbb H^3$,  then 
	\begin{align} 
\label{eq:160712.2}
	{\bm r}\big(v(t)\big) = {\bm r}  
	\big(v(0)\big)   e^t, \qquad t \geq 0.
\end{align}
\end{lem}
\begin{proof}

Fix $x \in \mathbb{H}^3$ and define $y_i := x_i - 1/3$ for each $i = 1,2,3$. Then we get
\begin{align*}
	{\bm r}(x) &= \frac{1}{3} \Big(\big(y_1 - y_2\big)^2 + \big(y_1 - y_3\big)^2 + \big(y_2 - y_3\big)^2\Big) = \frac{2}{3} \sum_{i =1}^3 y_i^2 - \frac{2}{3} \big (y_1 y_2 + y_1 y_3 + y_2 y_3\big)\\
	&= \frac{2}{3} \sum_{i =1}^3 y_i^2 + \frac{2}{3} \big (y_1^2 + y_2^2 + y_1 y_2\big) =  \frac{2}{3} \sum_{i =1}^3 y_i^2 + \frac{1}{3} \sum_{i =1}^3 y_i^2  = \sum_{i =1}^3 y_i^2
	= \sum_{i = 1}^3 \left(x_i-\frac{1}{3}\right)^2,
\end{align*}
 using      $y_3 = -y_1 - y_2$ repeatedly.     Basic stochastic calculus and \eqref{eq:160718.1}, yield now the very simple, deterministic dynamics $\dx {\bm r}(v(t)) = {\bm r}(v(t)) \,\dx t$ for all $t \geq 0$, provided that $v(0) \in \mathbb H^3$; and 
\eqref{eq:160712.2} follows.
\end{proof}

\begin{proof}[Proof of Proposition~\ref{P:160712.1}]
	We let $v(\cdot)$ denote the solution of the system of stochastic equations described in \eqref{eq:160512.4}--\eqref{eq:160512.6} for some $v(0) \in {\bm \Delta}^3_+$. Next, we define the stopping time 
		\begin{align*}
		\tau  := \inf \big\{t \geq 0: v(t) \notin \bm \Delta^3_+ \big\}
	\end{align*}
and the stopped  process $ \mu(\cdot) := v (\cdot \wedge \tau)$.
 This   is a vector of martingales, so relative arbitrage,  with respect to a market with the components of $\mu (\cdot)$ as its   relative weights,  is impossible, over any given time horizon $[0,T]$ with  $T>0$; see also Remark~\ref{R:4.2}.
 
 Now, the definition of the stopping time $\tau$ implies that   $\bm Q(\mu(\tau)) \leq 1/2$, thus also $\bm r(\mu(\tau)) \geq 1/6$, hold.  In conjunction with Lemma~\ref{L:160509.1}, this yields that $\tau$ is bounded away from zero, namely, that $ T^* (\mu (0)) := \log \big( 1 / ( 6 {\bm r} (\mu (0)) ) \big) \leq \tau$  holds,  since  $\bm Q(\mu(0)) >1/2 $.    Moreover,  with \eqref{160512.4} and \eqref{eq:160712.2} we have
 $$
 \frac{\partial}{\partial t} \Gamma^{\bm Q}(t) = {\bm r}  (\mu(t) )   \ge  {\bm r}  \big(\mu(0) \big), \qquad t \in [0, T^* (\mu (0))].
 $$  
Hence,  \eqref{eq: 4.3mod} is satisfied with $ \bm G=\bm Q$ and $\eta ={\bm r}  (\mu(0) ) = 2/3 - \bm Q (\mu(0))\in( 0, 1/6)$, thanks to \eqref{eq:160718.1}.
\end{proof}

\begin{rem}[A sanity check]
 \label{K:160813.1}
We can verify that  
 $T^* (\mu (0)) < \bm Q(\mu(0))/ \eta$
holds with the notation of the above proof, and in accordance with Corollary \ref{C:4.4}.
 \end{rem}
 
\begin{rem}[Expanding circle]
 \label{K:160813.2}
Let us observe from \eqref{eq:160712.2}, that the market weights constructed in Proposition~\ref{P:160712.1} live on an expanding circle.  More specifically, from \eqref{eq:160718.1}, \eqref{eq:160712.2}, and \eqref{eq:Q} we have  
$$
\mu_1^2 (t) + \mu_2^2 (t) +\mu_3^2 (t) = \frac{1}{3} + {\bm r}\big(\mu(t)\big) = \frac{1}{3} + {\bm r} \big(\mu(0)\big)   e^t, \qquad t \in [0, T^*];
$$
hence  the vector $  \mu (\cdot)   $ of relative market weights lies on the intersection of the hyperplane $\mathbb{H}^3$ with the sphere of radius  $  \sqrt{(1/3) + {\bm r}(\mu(0)) e^t }$ centered at the origin. This intersection is a circle of radius $\sqrt{ {\bm r}(\mu(0))  e^t}$ centered at the node $(1/3, 1/3, 1/3)'$.  
\end{rem}

\subsection{Starting away from the node, and ``moving slowly"}  
\label{SS:6.2}

As \eqref{eq:160712.2} shows, in the context of Proposition~\ref{P:160712.1}, the process $\mu (\cdot)$ starts out away from the node  on the lateral face of the unit simplex, then spins outward {\it very} {\it fast} (namely, exponentially fast),   until it reaches the boundary of the simplex at some time $  \tau \ge \log \big( 1 / ( 6 {\bm r} (\mu (0)) ) \big)$; this time is bounded away from zero.

We   construct here another   market model,   similar  to the one in Subsection~\ref{SS:6.1}, but in which the spinning motion of   $\mu (\cdot)$ is ``slowed down" quite considerably. More precisely, the diffusion process of Theorem~\ref{T:160517.1}  starts  away from the node $(1/3,1/3,1/3)',$   then moves  outwards along level sets of the quadratic  function $\bm Q$. This   takes time at least $T = \bm Q(\mu_0)-1/2$; on the interval $[0,T]$ the condition in \eqref{eq: 4.3mod} is  satisfied with $\bm G = \bm Q$ and $\eta =1$, but no arbitrage with respect to the market can exist.

\begin{thm}[Lack of short term relative arbitrage opportunities]
\label{T:160517.1}
Assume that the filtered probability space $(\Omega, \sigalg{F},   \P), $ $ \filt{F} = (\sigalg{F} (t) )_{t \geq 0 }$ supports a Brownian motion $W(\cdot)$.  Fix $\mu_0 \in \bm \Delta^3_+$  with $\bm Q(\mu_0) > 1/2$. Then there exists an  It\^o diffusion  $\mu (\cdot)= ( \mu_1 (\cdot  ), \mu_2 (\cdot  ), \mu_3 (\cdot  ))'$  with values in ${\bm \Delta}^3$,     a time-homogeneous  dispersion matrix,  starting point $\mu(0) = \mu_0$, and the following properties: 
\begin{enumerate}
[label={\rm(\roman{enumi})}, ref={\rm(\roman{enumi})}]
\item\label{T:160517.1:2} 
No relative  arbitrage exists with respect to the market with relative  weight process $\mu (\cdot)$,    over any time horizon $[0,T]$ with   $T > 0$.
\item\label{T:160517.1:1} The condition of \eqref{eq: 4.3mod} is satisfied for the quadratic function $\bm G=\bm Q$
with $\eta = 1$ and with $T = \bm Q(\mu_0)-1/2$.
\end{enumerate}
\end{thm}

We prove   this theorem at the end of the subsection, after a remark and a preliminary result.

\begin{rem}[An open question] \label{R:160719.1}
Suppose that the condition in \eqref{eq: 4.3} is satisfied by a market model with relative weight process $ \mu (\cdot)$, for the quadratic function $\bm Q$
with $\eta = 1$. 
	Theorem~\ref{T:4.3} yields   then the  existence of a strong relative arbitrage with respect to this market, over any time horizon $[0,T]$ with $T >\bm  Q(\mu(0))$.  
	On the other hand, Theorem~\ref{T:160517.1} shows that, for time horizons $[0,T]$ with $T \leq \bm Q(\mu(0)) - 1/2$, there exist  market models with respect to which no relative arbitrage is possible, even if \eqref{eq: 4.3} holds for them. 
	 We do not know what happens for time horizons $[0,T]$ with $T \in \big(\bm Q(\mu(0)) - 1/2,\bm Q(\mu(0)\big]$. We conjecture that relative arbitrage is possible for those time horizons, but that it need not be strong.
\end{rem}

For the next result, we recall the function ${\bm r}$ from \eqref{160512.4}.

\begin{prop}[Time-changed, slowed-down  version of \eqref{eq:160512.4}--\eqref{eq:160512.6}]  
\label{P:160514.1}
Assume that the filtered probability space $(\Omega, \sigalg{F},   \P)$, $\filt{F} = (\sigalg{F} (t) )_{t \geq 0 }$ supports a Brownian motion $W(\cdot)$.  
Then, for any initial condition $w(0) \in \mathbb{H}^3$ with $w(0) \neq(1/3,1/3,1/3)',$ the following system of stochastic differential equations has a pathwise unique strong solution $w(\cdot)$, taking values in $\,\mathbb H^3:$  
 \begin{align} 
	\dx w_1(t) &= \frac{1}{\sqrt{3   {\bm r}(w(t))}}  \big(w_2(t) - w_3(t)\big)\, \dx W(t);  \label{eq:160512.1}\\
	\dx w_2(t) &= \frac{1}{\sqrt{3   {\bm r}(w(t))}}  \big(w_3(t) -w_1(t)\big)\, \dx W(t);  \label{eq:160512.2}\\
	\dx w_3(t) &= \frac{1} {\sqrt{3  {\bm r}(w(t))}}   \big(w_1(t) -w_2(t)\big)\, \dx W(t).  \label{eq:160512.3}
\end{align}
\end{prop}
\begin{proof}
Let $w(\cdot)$ denote any solution to the system \eqref{eq:160512.1}--\eqref{eq:160512.3} with  $w(0) \in \mathbb{H}^3 \setminus \{(1/3,1/3,1/3)'\}$. Then it is clear that  $w(t) \in \mathbb H^3$ holds   for all $t \ge 0$. Next, we define the stopping time 
\begin{align*}
	\sigma = \inf \left\{t \geq 0: {\bm r}(w(t)) < \frac{{\bm r}(w(0))}{2}\right\} 
\end{align*}
and note as in Lemma~\ref{L:160509.1}, via an application of It\^o's formula,  that
\begin{align*}
	\dx {\bm r} (w(\sigma \wedge t)) =  \1_{\{\sigma>t\}} \,\dx t, \qquad t \geq 0. 
\end{align*}
Hence, if $\sigma>0$,     the process ${\bm r}(w(\cdot))$ is nondecreasing and deterministic; indeed,  we then have $ {\bm r}(w(t)) = {\bm r}(w(0)) + t$ for all $t \geq 0$, and $\sigma = \infty$. For this reason, given any $\varepsilon \in (0, 3 {\bm r}(w(0))$, any solution to the system \eqref{eq:160512.1}--\eqref{eq:160512.3} solves also the system
\begin{align} 
	\dx w^\varepsilon_1(t) &= \frac{1}{\sqrt{\varepsilon \vee 3 {\bm r}(w^\varepsilon(t))}}  \big(w^\varepsilon_2(t) - w^\varepsilon_3(t)\big) \, \dx W(t),   \label{eq:160512.1'}\\
	\dx w^\varepsilon_2(t) &= \frac{1}{\sqrt{\varepsilon \vee 3 {\bm r}(w^\varepsilon(t))}}  \big(w^\varepsilon_3(t) - w^\varepsilon_1(t)\big) \, \dx W(t),  \label{eq:160512.2'}\\
	\dx w^\varepsilon_3(t) &= \frac{1} {\sqrt{\varepsilon \vee 3 {\bm r}(w^\varepsilon(t))}} \big (w^\varepsilon_1(t) - w^\varepsilon_2(t)\big)\, \dx W(t).  \label{eq:160512.3'}
\end{align}
Since the system \eqref{eq:160512.1'}--\eqref{eq:160512.3'} has Lipschitz-continuous coefficients, its solution is unique. This yields uniqueness of the solution to the system \eqref{eq:160512.1}--\eqref{eq:160512.3}. Existence of a solution to the system \eqref{eq:160512.1}--\eqref{eq:160512.3} follows by checking that any solution to \eqref{eq:160512.1'}--\eqref{eq:160512.3'}  is also a solution to  \eqref{eq:160512.1}--\eqref{eq:160512.3}.
\end{proof}

\begin{proof}[Proof of Theorem~\ref{T:160517.1}]
We proceed exactly as in the proof of Proposition~\ref{P:160712.1}. With $w(0)  = \mu_0$, we recall the process $w(\cdot)$ of Proposition~\ref{P:160514.1}, define the stopping time 
\begin{align*}
		\tau   :=   \inf \big\{t \geq 0: w(t) \notin \bm \Delta^3_+ \big\},
	\end{align*}
	 	 and set $\mu(\cdot):= w (\cdot \wedge \tau)$. 
		 This process $\mu(\cdot)$ is a vector of continuous martingales; hence, no relative arbitrage is possible with respect to this market, over {\it any} given time horizon.
Finally, we note that $\Gamma^{\bm Q}(t) = t \wedge \tau $ holds for all $t \geq 0$, and that $ \bm Q(\mu_0)- \tau =\bm Q(\mu(\tau)) \leq 1/2$.  This then yields the statement.
\end{proof}

\subsection{Modifications}   \label{SS:6.3} 

We provide here modifications of the examples in the previous subsections. More precisely, let  us recall the first system of stochastic differential equations \eqref{eq:160512.4}--\eqref{eq:160512.6} along with the solution given in \eqref{eq:160712.3}. That is, let us consider, for some $\delta \in (0,1/3)$ and $u\in \R$, the corresponding market model 
\begin{align*} 
		\mu_i(t) = \frac{1}{3} + \delta  e^{t/2} \cos \left(W(t) + 2 \pi \left( u   + \frac{i-1}{3}  \right) \right), \qquad i  = 1,2,3,\,\, t \in [0, -2 \log(3\delta)],
	\end{align*}
	where $W(\cdot)$ is Brownian motion and $u$ a real number. 
 
The first modification perturbs this model radially in the plane of ${\bm \Delta}^3,$ so that the resulting new model has a covariation matrix-valued process   $\alpha(\cdot)$ as in \eqref{eq:160524.1}  with two strictly positive eigenvalues. 

\begin{prop}[Nondegeneracy and absence of arbitrage]
\label{P:160712.2}
Assume that the filtered probability space $(\Omega, \sigalg{F},   \P)$, $\filt{F} = (\sigalg{F} (t) )_{t \geq 0 }$  supports two independent Brownian motions $W(\cdot)$ and $B(\cdot)$.  
Then there exist a real number $T^*>0$ and an It\^o diffusion  $\mu (\cdot)= ( \mu_1 (\cdot  ), \mu_2 (\cdot  ), \mu_3 (\cdot  ))',$ with the following properties:
\begin{enumerate}[label={\rm(\roman{enumi})}, ref={\rm(\roman{enumi})}]
\item  No relative arbitrage with respect to the market $\mu(\cdot)$  exists,      over any time horizon $[0,T]$ with   $T > 0$. 
\item\label{P:160712.2:2}  The condition of \eqref{eq: 4.3mod} is satisfied with  $\bm G=  \bm Q$, given in \eqref{eq:Q}, 
  $\eta = \bm r(\mu(0))/4$,  and  $T=T^*$. 
\item  The  covariation matrix of $\mu (\cdot)$ has two strictly positive eigenvalues on $[0,T^*];$ that is, for each $t \in [0,T^*]$ the matrix $\alpha(t)$ of \eqref{eq:160524.1} has two strictly positive eigenvalues.
\end{enumerate}
\end{prop} 

\begin{proof}
To describe the model, let us fix a real constant  $\delta \in (0, 1/9)$. We consider also  a martingale $\Phi (\cdot) $ of the filtration generated by the Brownian motion $B(\cdot)$ with values in the interval $(\delta, 3\delta)$ and starting point  $\Phi(0) = 2\delta$. More precisely, we introduce the It\^o diffusion
 $$
 \Psi (\cdot) := \int_0^\cdot \big(\Psi (t) - \delta \big)  \big(\Psi (t) + \delta \big)\, \dx B(t)
 $$ 
 with state space  $(-\delta, \delta);$
 we note that $\langle \Psi\rangle(\cdot)$ is strictly increasing thanks to Feller's test of explosions; and  define the martingale $\Phi (\cdot) := 2 \delta + \Psi (\cdot)$  with   takes values in the interval $(\delta, 3\delta)$.

We now define $T^* :=  -2 \log(9\delta)>0$ and the market weights as positive It\^o processes 
\begin{equation} \label{eq:160804}
\mu_i (t) :=  \frac{1}{3} + \Phi ( t \wedge T^*) e^{(t \wedge T^*)/2}  \cos \Big(  W (t \wedge T^*)  + \frac{2 \pi}{3}   \big(i -1\big) \Big), \qquad  i  = 1,2,3,\,\, t \geq 0.
\end{equation}
Since $\Phi(\cdot)$ and $W(\cdot)$ are independent, $\mu_i(\cdot)$ is a martingale for each $i = 1, \cdots, 3$. Indeed, we have the dynamics
 $$
 \dx \mu_i (t) = - \Phi ( t)  e^{t/2} \sin   \Big(  W (t)  + \frac{ 2 \pi}{3}  \big(i -1\big) \Big) \, \dx W (t) + e^{ t/2} \cos   \Big(  W (t)  + \frac{2 \pi}{3}  \big(i -1\big) \Big) \, \dx \Phi (t)  
 $$
 for all $i  = 1,2,3$ and $t \in [0, T^*]$.
As a result, no relative arbitrage can exist with respect to the market. Moreover, we note 
\begin{align*}
	\langle \mu_i \rangle (\cdot) \geq \int_0^{\cdot \wedge T^*} \Phi^2 ( t)  e^t \sin^2   \Big(  W (t)  + \frac{ 2 \pi}{3}  \big(i -1\big) \Big) \dx t \geq \frac{\delta^2}{4}    \int_0^{\cdot \wedge T^*}  e^t \sin^2   \Big(  W (t)  + \frac{ 2 \pi}{3}  \big(i -1\big) \Big) \dx t.
\end{align*}
Recall that $\Gamma^{\bm Q}(\cdot) = \sum_{i = 1}^3 \langle	\mu_i \rangle (\cdot)$; hence we get
\begin{align*}
	\frac{\partial}{\partial t} \Gamma^{\bm Q}(t) \geq \frac{\delta^2}{4} \sum_{i=1}^3
		  \sin^2   \Big(  W (t)  + \frac{ 2 \pi}{3}  \big(i -1\big) \Big)  = \frac{{\bm r}(\mu(0))}{4}, \qquad t \in [0, T^*].
\end{align*}
yielding \ref{P:160712.2:2}. Here, the last equality follows from the same type of computations as the ones in Remark~\ref{R:160725.1}.

Of the two Brownian motions that drive this market model, $W (\cdot)$ generates circular motion in the plane of ${\bm \Delta}^3$ about  the point $(1/3, 1/3, 1/3)'$; while $B(\cdot) $  generates the martingale $\Phi (\cdot)$, whose quadratic variation has strictly positive time-derivative. Thus these two independent random motions span the two-dimensional space, and the covariation  process $\alpha (\cdot)$ of the market weight process $\mu (\cdot)$ has rank   $2$. 
\end{proof}

\begin{rem}[Contrasting Theorem~\ref{T:5.2} to Proposition~\ref{P:160712.2}]  \label{R:P:160712.2}
Theorem~\ref{T:5.2} yields the existence of short-term arbitrage, if the $(d-1)$ largest eigenvalues of the covariance matrix $\alpha(\cdot)$ are  uniformly bounded away from zero.  According to  Proposition~\ref{P:160712.2}      the weaker requirement, that  all these   $(d-1)$ largest eigenvalues be strictly positive, is {\it not} sufficient to guarantee short-term relative arbitrage.
Indeed, the slope of $\langle \Psi\rangle(\cdot)$ in the proof of Proposition~\ref{P:160712.2}  can get arbitrarily close to zero over any given time horizon, with positive probability. This prevents the second-largest eigenvalue of the covariance matrix $\alpha(\cdot)$ from being uniformly bounded from below, away from zero.   
 \end{rem} 
     
\begin{rem}[Spiral expansion]  \label{R:P:010816}
 In the market model of \eqref{eq:160804}, the relative weights take values on a circle that is allowed to expand at the rate $\Phi (t) e^{t/2}$;  
 more precisely, we have   
$$
\mu_1^2 (t) + \mu_2^2 (t) +\mu_3^2 (t)  = \frac{1}{3} + \frac{3 \Phi^2 (t) e^{t}}{2}  < \frac{1}{3}, \qquad t \in [0,T^*].
$$
Therefore, at any given time $ t \in [0,T^*]$  the vector $ \mu ( t) $ of relative market weights lies on the intersection of the hyperplane $\mathbb{H}^3$ with the sphere of radius  $  \sqrt{(1/3) + (3/2)  \Phi^2 (t) e^t }$ centered at the origin. This intersection is a circle of radius   $  \sqrt{   3/2    } \, \Phi (t) \,e^{t/2}< \sqrt{1/6}\,,$     contained in ${\bm \Delta}^3_+$ and centered at its node $(1/3, 1/3,1/3)'$; see also Remark~\ref{K:160813.2}.  Thus, a more precise description of the current situation might be that  the market weights live on a spiral. The rate of this ``spiral expansion" is exactly   that one, for which the market weights become  martingales. 
  \end{rem} 
 
 This remark raises the following question: {\it What happens if   the market weights are confined to a stationary circle in  ${\bm \Delta}^3$?} Such a diffusion, confined to a circle, is incompatible with a martingale structure.  This is the subject of the   example that follows.

 \begin{example}[Immediate arbitrage]  \label{Ex:immediate}
 Let $W (\cdot)$ denote a  scalar Brownian motion,  fix a real constant $\delta \in (0, 1/3)$, and    define the positive  market weight processes
 \begin{align*}
 \mu_i (\cdot) :=   \frac{1}{3} + \delta \cos \Big(  W (\cdot)  + \frac{2 \pi}{3}  \big(i -1\big) \Big), \qquad i = 1, 2, 3.
 \end{align*}
These   market weights 
take values in the interval $(0, 2/3)$; in fact they live on a circle, namely 
$$
\mu_1^2 (t) + \mu_2^2 (t) +\mu_3^2 (t)= \frac{1}{3} + \frac{3  \delta^2}{2}  < \frac{1}{2},\qquad \text{so} \qquad {\bm r}\big(\mu(\cdot)\big)\equiv  \frac{3 \delta^2}{2} 
$$
in the notation of \eqref{160512.4},  and have dynamics
 $$
 \dx \mu_i (t) =  - \delta \sin   \Big(  W (t)  + \frac{2 \pi}{3}  \big(i -1\big) \Big) \dx W (t) - \frac{\delta}{2}  \cos \Big(  W (t)  + \frac{ 2 \pi}{3}   \big(i -1\big) \Big)  \dx t, \qquad t \geq 0.
 $$
 As above,   we have
 \begin{align*}
 	\Gamma^{\bm Q}(t) = \int_0^t {\bm r}\big(\mu(s)\big) \,\dx s =   {\bm r}\big(\mu(0)\big) t =  \frac{3  \delta^2 t}{2},  \qquad t \geq 0.
 \end{align*}
 
Let us introduce now the normalized quadratic function ${\bm Q}^\star:= {\bm Q} /  \bm Q \big( \mu (0) \big)$,  where   $\bm Q \big( \mu (0) \big) = (2/3) - \big( 3 \delta^2 / 2\big) >0$.  For the   trading strategy ${\bm \varphi}^{{\bm Q}^\star}( \cdot)$ generated additively by this function  ${\bm Q}^\star$ as in \eqref{eq:3.2}, and the associated wealth process $V^{{\bm \varphi}^{{\bm Q}^\star}}(\cdot)$ of  \eqref{eq: valphiG}, we   get 
\begin{equation*}
V^{{\bm \varphi}^{{\bm Q}^\star}} (t ) = {\bm Q}^\star \big( \mu (t) \big) + \Gamma^{{\bm Q}^\star} (t)
= 1 + \frac{3 \delta^2}{2 \bm Q  ( \mu (0)  )} \,t, \qquad t \geq 0.
\end{equation*}
Hence, the additively-generated strategy ${\bm \varphi}^{{\bm Q}^\star}( \cdot)$ yields a strong relative arbitrage over any given time horizon $[0,T]$.  Investing according to this  strategy  ${\bm \varphi}^{{\bm Q}^\star}( \cdot)$  is a sure way   to do better than    the market right away, and to keep doing better and better as time goes on.

Indeed, the existence of such ``egregious", or ``immediate", arbitrage, should not come here as a surprise,  since the so called ``structure equation'' is not satisfied;   
in particular, no deflator as in Definition~\ref{D:deflator} exists for $\mu(\cdot)$; see, for example,  \cite{Schweizer_1992}, or Theorem~1.4.2 in \cite{KS2}.
 \qed
 \end{example}

\begin{rem}[General submanifolds]\label{rem: 6.3} The diffusion constructed In Example~\ref{Ex:immediate}  lives on a submanifold of $\R^3$, which is incompatible with a martingale structure. By contrast, in Subsection~\ref{SS:6.1} the submanifold was allowed to evolve as an expanding circle.
 The diffusions in these examples could probably be generalized to   diffusions with support on an arbitrary submanifold of $\R^d$, for $d\ge 2$, and then the submanifold could  be allowed to evolve through $\R^d$ in the manner of our expanding circle in  $\R^3$. 
 
 In this case,  a natural question would be to characterize the evolution that would cause the diffusion on the evolving submanifold to become a martingale. How would this martingale-compatible evolution depend on the diffusion? What would become of this evolving submanifold over time? Would the evolving submanifold develop singularities? Et cetera. 
 We provide some partial answers to such questions in the following Subsection~\ref{SS:6.4}, but the picture that emerges  seems to be far from complete.
\end{rem}

\subsection{General Lyapunov functions} 
\label{SS:6.4}

We   extend now the construction of market models, for which  \eqref{eq: 4.3mod} is satisfied  but no short-term arbitrage is possible, to a general class of generating functions. Throughout this section we fix a  Lyapunov function $\bm G: \bm \Delta^3 \to [0, \infty)$; this function is assumed to be strictly concave, and twice continuously differentiable in a neighborhood of $\bm \Delta^3_+$. We    assume moreover that  its Hessian $D^2{\bm G}$ is locally Lipschitz continuous. Next, we introduce  the nonnegative   number 
\begin{align*}
	\mathfrak{ g} := \sup_{x \in \bm \Delta^3 \setminus \bm \Delta^3_+}\bm G(x).
\end{align*}
If  $\bm G$  attains an interior local (hence also global) maximum at a point $c \in \bm \Delta^3_+$, we call this  $c$ the ``navel'', or ``umbilical point", of the function $\bm G$.  

\begin{thm}[General Lyapunov functions, lack of short-term relative arbitrage]
\label{thm: No_Arb_Suff_Short}
 Assume that the filtered probability space $(\Omega, \sigalg{F},   \P)$, $\filt{F} = (\sigalg{F} (t) )_{t \geq 0 }$ supports a Brownian motion $W(\cdot)$.     Suppose also that we are given a Lyapunov function $\bm G$ with the properties and notation just stated, along with a vector  $\mu_0 \in {\bm \Delta}^3_+ $ such that $\bm G(\mu_0) \in \big(\mathfrak g,  \max_{x \in \bm \Delta^3}\bm G(x)\big)$. 
Then there exists an {It\^o} diffusion  $\mu (\cdot) = ( \mu_1 (\cdot), \mu_2(\cdot), \mu_3 (\cdot) )'$ with starting point $\mu (0) =\mu_0$, values in ${\bm \Delta}^3$, and the following properties:
\begin{enumerate}[label={\rm(\roman{enumi})}, ref={\rm(\roman{enumi})}]
\item No relative  arbitrage is  possible with respect to $\mu(\cdot)$  over any time horizon $[0,T]$  with  $T > 0$.
\item The condition of \eqref{eq: 4.3mod} is satisfied with $\eta = 1$ and $T = \bm G(\mu_0)-\mathfrak{ g}$.    
\end{enumerate}
Moreover, we have
$$
{\bm G}\big(\mu (t) \big) = \bm G(\mu_0) -  t; \qquad \Gamma^{\bm G} (t) = t  \qquad t \in [0,T].
$$
\end{thm}

\begin{proof}
We introduce the vector function $\sigma  = ( \sigma_1, \sigma_2, \sigma_3)'$ with components  
\begin{equation*}
\sigma_1 (x) :=  D_3 {\bm G}(x) - D_2 {\bm G}(x); \quad \sigma_2 (x) :=  D_1 {\bm G}(x) - D_3 {\bm G}(x);  \quad \sigma_3 (x) :=  D_2 {\bm G}(x) - D_1 {\bm G}(x), \quad x \in {\bm \Delta}^3_+.
\end{equation*}
If $\sigma_1 (x) = \sigma_2(x) = \sigma_3(x) = 0$, then $x = c$ is the umbilical point.   Indeed, if for some $x \in \bm \Delta^3_+$ we have $D_1 {\bm G}(x) = D_2 {\bm G}(x) = D_3 {\bm G}(x)$,  then the strict concavity of $\bm G$ yields
$$
0 =  \sum_{i=1}^3 D_i {\bm G}(x)  \big( x_i - y_i \big) > \bm G(y) - \bm G(x), \qquad  y \in {\bm \Delta}^3_+ .
$$
Next, we introduce the function 
\begin{equation}
\label{eq: Non_Deg_Sigma}
L(x) := - \frac{1}{2}  \sigma^\prime (x)  D^2{\bm G} (x) \sigma  (x), \qquad x \in {\bm \Delta}^3_+ 
\end{equation}
and note that   $L(x) > 0$ holds, as long as $x$ is not the umbilical point $c$.

Let us now consider the {It\^o} diffusion process   $ \mu (\cdot) =  ( \mu_1 (\cdot), \mu_2(\cdot), \mu_3 (\cdot)  )'$ with initial condition $\mu (0) = \mu_0$ and dynamics
\begin{equation}
\label{eq: circ_sym}
\dx \mu_i (t) = \frac{\sigma_i \big( \mu (t) \big)}{\sqrt{L (\mu (t))}} \, \dx W(t), \qquad i=1, 2,3.
\end{equation}
 It is clear from this construction and the property $ \sum_{i=1}^3 \sigma_i (x)= 0$ for all $x \in  {\bm \Delta}^n_+$ that, as long as the process $\mu (\cdot)$ is well defined by the above system, it satisfies $\mu_1 (\cdot)+  \mu_2(\cdot) + \mu_3 (\cdot) = 1$. Furthermore, if the process $\mu (\cdot)$ is well defined, elementary stochastic calculus and the property $$ 
\sum_{i=1}^3 D_i {\bm G} (x)   \sigma_i (x) = 0, \qquad x \in  {\bm \Delta}^n_+ 
$$ 
lead to the dynamics
\begin{align*}
\dx {\bm G}\big(\mu (t) \big) = \frac{1}{2} \sum_{i=1}^3 \sum_{j=1}^3 D^2_{ij}{\bm G}\big(\mu (t) \big)\, \dx \langle \mu_i , \mu_j \rangle (t) = - \dx \Gamma^{\bm G}(t).
\end{align*}
This last double summation is identically equal to $-1$, by virtue of \eqref{eq: Non_Deg_Sigma}.  
Thus, $\bm G(\mu(\cdot))$ is decreasing, and stays clear of the navel $c$ (if  this exists). Hence, by analogy with the proof of Proposition~\ref{P:160514.1}, we get that the process $\mu(\cdot)$ is well defined since the system of stochastic differential equations in \eqref{eq: circ_sym} has a pathwise unique, strong solution, up until the time $ \mathscr{D}$. This stopping time was defined in \eqref{eq: 6} and describes the first time when $\mu(\cdot)$ hits the boundary of $\bm \Delta^3$; in particular, we have for it
\begin{equation}
\label{eq: bounds_Delta}
\bm G(\mu_0)-\mathfrak{ g}  \le  \mathscr{D} = G(\mu_0) - G\big(\mu(\mathscr{D})\big)  \le \bm G(\mu_0).
\end{equation}
The market weight processes $ \mu_i (\cdot)$  for $i = 1,2,3$ are continuous martingales with values in the unit interval $[0,1]$. As a result, no relative arbitrage can exist with respect to the resulting market, over any given time horizon. Since $ \mathscr{D} \geq \bm G(\mu_0)-\mathfrak{ g}$,  we can conclude the proof.
\end{proof}

Theorem~\ref{T:160517.1} is a special case of Theorem~\ref{thm: No_Arb_Suff_Short}. We do not know whether it is possible to remove the assumption $\bm G(\mu_0) \neq \max_{x \in \bm \Delta^3} \bm G(x)$ from Theorem~\ref{thm: No_Arb_Suff_Short}, but conjecture that  this  should be possible.

\begin{rem}[Presence of a gap] 
\label{R:gap}
We continue here the discussion started in Remark~\ref{R:160719.1}. It is very instructive to compare the interval $ \big( \bm G(\mu (0) , \infty \big) $ of \eqref{eq: 4.4},    which provides the lengths of time horizons $[0,T]$ over which strong arbitrage is possible with respect to {\it any} market  that satisfies the condition in \eqref{eq: 4.3} with $\eta =1$; and the interval $\big[0, \bm G(\mu (0) )-  \mathfrak{ g}  \big] $ of Theorem~\ref{thm: No_Arb_Suff_Short},      giving the lengths of time horizons $[0,T]$ over which examples of markets  can be constructed that do not admit  relative arbitrage but satisfy the condition in \eqref{eq: 4.3}.

There is a  gap in these two intervals, when $  \mathfrak{ g}$ is positive. For instance, in the case of the entropy function $\bm G=\bm H$ 
of \eqref{eq: 3.13}, the gap is very much there, as we have $\mathfrak{ g} =2  \log  2 $, but  $\max_{x \in \bm \Delta^3} \bm H(x) = 3 \log 3$.  In this ``entropic" case, the dynamics of \eqref{eq: circ_sym} and \eqref{eq: Non_Deg_Sigma} take the form
\begin{align*}
\dx \mu_i (t) &= \log \bigg( \frac{\mu_{i+1} (t)}{\mu_{i-1} (t)}  \bigg) \frac{\dx W(t)}{ \sqrt{ L (\mu (t)) } }  ,\quad i = 1, 2,3 ,\,\,t \geq 0, \quad \text{with } \mu_0(\cdot) := \mu_3(\cdot),  \, \mu_4(\cdot) := \mu_1(\cdot);\\
	L(x) &= \sum_{i=1}^3  \frac{1}{2 x_i} \left( \log \Big( \frac{x_{i+1}}{x_{i-1}}\Big) \right)^2,  \qquad x \in {\bm \Delta}^3_+ , \quad \text{with } ~x_0 := x_3, \, x_4 := x_1.
\end{align*}
\end{rem}
 
\begin{rem}[Absence of a gap] 
\label{R:nogap}
The gap in Remark \ref{R:gap} disappears, of course, when $\mathfrak{ g} =0$;   most eminently,  when the  function $\bm G$ vanishes on the boundary  $\bm \Delta^d \setminus \bm \Delta^d_+$,  but is strictly positive on $\Delta^d_+$. 

For such a function $\bm G$ and for the market model constructed in the proof of Theorem~\ref{thm: No_Arb_Suff_Short}, we are then led from \eqref{eq: bounds_Delta} and with   the notation of \eqref{eq: 6},  to the rather remarkable identity 
$$
 \mathscr D = \mathscr D_1 \wedge \mathscr D_2 \wedge \mathscr D_3= \bm G(\mu (0)).
$$
To wit: at exactly the time $T = \bm G(\mu_0)$, one of the components of the diffusion $ \mu (\cdot)$ vanishes for the first time.

A prominent example of such a function is the geometric mean $\bm R$ from \eqref{eq:R};
 in this case, and   with $\mu_0(\cdot)$, $\mu_4(\cdot)$, $x_0$, and $x_4$ as above, the dynamics corresponding to \eqref{eq: circ_sym} and \eqref{eq: Non_Deg_Sigma} are given as 
 $$
 \dx \mu_i (t) =   \bigg( \frac{1}{\mu_{i-1} (t)} - \frac{1}{\mu_{i+1} (t)}   \bigg) \frac{ \dx W(t)}{ \sqrt{ L_* (\mu (t)) } } ,\qquad i = 1, 2,3,\,\, t \geq 0 .
 $$
Here we have set 
\begin{align*}
L_*(x) &=  \frac{\bm R(x)}{18} \left(
\sum_{i=1}^3 \frac{3}{  x_i^2} \left(     \frac{1}{x_{i-1}} -  \frac{1}{x_{i+1}}  \right)^2 - \sum_{i=1}^3  \sum_{j=1}^3   \frac{1}{ x_i  x_j } \left(     \frac{1}{x_{i-1}} -  \frac{1}{x_{i+1}}  \right) \left(     \frac{1}{x_{j-1}} -  \frac{1}{x_{j+1}}  \right)     \right)\\
	&= \frac{\bm R(x)}{6 x_1^2 x_2^2 x_3^2} 
\sum_{i=1}^3  \big( {x_{i+1}} -  {x_{i-1}}  \big)^2 = \frac{1}{2 \big( \bm R(x)\big)^5} 
 \sum_{i=1}^3  \bigg( x_i  -\frac{1}{3} \bigg)^2= \frac{\bm r (x)}{2 \big( \bm R(x)\big)^5}, \qquad x \in \bm \Delta^3_+;
\end{align*}
 the next-to-last equality follows  from \eqref{eq:160718.1}, and the last from \eqref{eq:160718.1}. 
\end{rem}

\section{Summary}\label{S:7}
In this paper we   place  ourselves in the context of continuous semimartingales $\mu(\cdot)$ taking values  in the $d$--dimensional simplex $\bm \Delta^d$. Each component of $\mu(\cdot)$ is interpreted as the relative capitalization of a company in an equity market, with respect to the whole market capitalization. We then study conditions on the volatility structure of $\mu(\cdot)$ that guarantee the existence of relative arbitrage opportunities with respect to the market. More precisely, we consider conditions that bound a cumulative  aggregation $\Gamma^{\bm G}(\cdot)$ of the volatilities of the individual components of $\mu(\cdot)$ from below. Here, the aggregation is done according to a so-called generating function $\bm G: \bm \Delta^d \to \R$, assumed to be sufficiently smooth. Then $\Gamma^{\bm G}(\cdot)$ is given by \eqref{eq: Gamma}, namely, 
 \begin{align*}
 	\Gamma^{\bm G} (\cdot)= -
\frac{1}{2}  \sum_{i=1}^d \sum_{j=1}^d \int_0^{ \cdot}    D^2_{i,j}
{\bm G}\big(\mu (t)  \big)   \,  \mathrm{d} \big \langle \mu_i, \mu_j  \big \rangle (t),
     \end{align*}
 and the condition on the cumulative aggregated market volatility   by \eqref{eq: 4.3}:
\begin{equation*}
\P \left( \text{the mapping} ~ [0,\infty) \ni t \mapsto \Gamma^{\bm G}(t)- \eta t \,\,\,\text{is nondecreasing} \right)= 1, \quad \text{for some  $\eta >0$}.
\end{equation*}
Section~\ref{S:4} recalls the trading strategy ${\bm \varphi}^{\bm G}(\cdot)$, which is a strong relative arbitrage   on the time horizon $[0,T]$ for all $T > \bm G(\mu(0)) / \eta$, provided that    \eqref{eq: 4.3} holds. It is important to note that this strategy ${\bm \varphi}^{\bm G}(\cdot)$ is ``model-free": it does not depend on the   specifications of a  particular model, and works for any continuous semimartingale model that satisfies \eqref{eq: 4.3}. 
 
 Section~\ref{S:5} provides several sufficient conditions   guaranteeing the existence of short-term relative arbitrage.  First, Subsection~\ref{SS:5.1} studies the question of strong short-term relative arbitrage.
 If a specific stock contributes to the overall market volatility, then Theorem~\ref{P:160715.1} yields the existence of such a strong arbitrage opportunity. The corresponding trading strategy turns out to be independent of the model specification but will be dependent on the choice of time horizon.  A similar construction is the underlying idea of \cite{Banner_Fernholz}. There, not a specific stock, but always the smallest one in terms of capitalization, contributes to the overall market volatility; see also Remark~\ref{R:160718.1}. 

Subsection~\ref{SS:5.2} yields sufficient conditions for the existence of short-term relative arbitrage, not necessarily strong. 
 The first sufficient condition concerns the support of $\mu(\cdot)$ and assumes that it is, in a certain weak sense, time-homogeneous (Theorem~\ref{T:5.1}). The second sufficient condition is on the strict nondegeneracy of the covariance process of $\mu(\cdot)$ (Theorem~\ref{T:5.2}). Any of the two conditions yields, in conjunction with \eqref{eq: 4.3}, the existence of a relative arbitrage opportunity for the time horizon $[0,T]$, for any $T>0$. However, the corresponding trading strategies usually depend on the model specification (including the specification on drifts) and also  on the time horizon $T$. Moreover, these sufficient conditions usually do not yield strong relative arbitrage. Nevertheless, if the market is complete, the strategies can be chosen to be strong relative arbitrages; see Remark~\ref{R:completeness}.

Section~\ref{S:6} answers negatively the long-standing question, whether the condition in \eqref{eq: 4.3} yields indeed the existence of relative arbitrage opportunities with respect to the market on $[0,T],$ for any given real number $T>0$. Specific counterexamples are constructed. In this section we assume that $d=3$, namely, that we are in a market of only three stocks. This is done   for notational convenience, as  a lower-dimensional market can always be embedded in a higher-dimensional one. A smaller choice for $d$   is not possible, as Proposition~\ref{C:160610.1} yields that in the case $d=2$ of two stocks the condition of \eqref{eq: 4.3} always yields the existence of relative arbitrage opportunities over any time horizon.

Section~\ref{S:6}  differs from the earlier  ones. In Section~\ref{S:5} we considered a fixed market model $\mu(\cdot)$  and formulated conditions on the market model that yield (possibly strong) relative arbitrage opportunities for short-time horizons. These trading strategies might or might not depend on the specific characteristics of $\mu(\cdot)$. Now, in Section~\ref{S:6}, we fix the generating function $\bm G$ and then construct market models $\mu(\cdot)$ for which there exist no relative arbitrage opportunities over time horizon $[0,T]$ for some $T>0$, but \eqref{eq: 4.3} holds. To be precise, the models we construct do not satisfy exactly \eqref{eq: 4.3}  but only the local (in time) version \eqref{eq: 4.3mod}. As pointed out in Remark~\ref{R:6.1}, this is done for notational convenience only. Some additional technical assumptions are made on the generating function $\bm G$ in Section~\ref{S:6}, most importantly that $\bm G$ is strictly concave and its second derivative is locally Lipschitz continuous. 

The constructed market models $\mu(\cdot)$ not only prevent relative arbitrage  and satisfy \eqref{eq: 4.3mod}, but also yield that $\bm G(\mu(\cdot))$ is a deterministic function of time. That is, $\mu(\cdot)$ {\it flows along the level sets of the Lyapunov function} $\bm G$.  

Moreover, Proposition~\ref{P:160712.2} yields the existence of a market model, which satisfies the nondegeneracy conditions of Theorem~\ref{T:5.2}  but {\it not   strictly,} and does {\it not} allow for relative arbitrage.   This shows that the conditions on $\mu(\cdot)$ in Theorem~\ref{T:5.2} are tight; see also Remark~\ref{R:P:160712.2}.

\medskip
While this paper answers some old open questions, it suggests several new ones. The three most important, in our opinion, are the following:
\begin{enumerate}
	\item Under \eqref{eq: 4.3} with $\eta =1$,  strong  relative arbitrage opportunities exist on $[0,T]$ for all $T> \bm G(\mu(0))$. In Section~\ref{S:6}, market models $\mu(\cdot)$ are constructed that satisfy \eqref{eq: 4.3} but do not admit relative arbitrage opportunities on $[0,T^*]$, for some real number $T^* = T^* (\mu (0))\in \big(0, \bm G(\mu(0)) \big) $. What can be said for time horizons $[0,T]$ with  $T \in (T^*, \bm G(\mu(0)]$? In this connection, see also Remarks~\ref{K:160813.1}, \ref{R:160719.1}, and \ref{R:gap}.
	\item The following question arises from the methodology used in Section~\ref{S:6} to construct the counterexamples. Assume a diffusion lives on a submanifold of $\R^d$, which is incompatible with a martingale structure (e.g., on a sphere).  If we now want to turn the diffusion into a martingale, how will the submanifold evolve through $\R^d$ (e.g., it could turn into an expanding sphere)? In this connection, see also Remark~\ref{rem: 6.3}.
	\item 
	Section~\ref{S:6} contains examples for diffusions, which could be turned into market models where short-term relative arbitrage is impossible but long-term relative arbitrage is possible. There are no additional frictions (e.g., trading costs, etc.) necessary to achieve this. Are there any interesting economic implications and, if so,  what are they? For instance, can such models arise from equilibrium theory in an economy where agents have preferences with respect to different time horizons?
\end{enumerate} 
 
 \medskip

\bibliography{aa_bib}{}
\bibliographystyle{apalike}

 \end{document}